\documentclass[11pt]{article}
\usepackage[bottom=1in,left=1in, right=1in, top=1in]{geometry}
\usepackage{tabularx}
\usepackage{enumitem}
\usepackage{graphicx}
\usepackage{wrapfig}
\usepackage{subfigure}
\usepackage{amsmath}
\usepackage{amsfonts}
\usepackage{cite}
\usepackage{hyperref}
\hypersetup{
	colorlinks=true, 
	linktoc=all,     
	linkcolor=blue,  
	citecolor = black
}

\usepackage[normalem]{ulem}

\newcommand{\cbcut}{\textsc{CB Hyper-st-Cut}}
\newcommand{\cbw}{\textsc{CBcut}$(r,\{w_1, w_2, \hdots, w_q\})$}
\newcommand{\cbf}{\textsc{CBcut}$(4,\{1, w_2\})$}
\usepackage{amsthm}
\usepackage{algorithm}
\usepackage[noend]{algorithmic}
\theoremstyle{plain}
\newtheorem{theorem} {Theorem} [section]
\newtheorem{lemma} [theorem] {Lemma}
\newtheorem{corollary} [theorem] {Corollary}
\theoremstyle{definition}
\newtheorem{definition}{Definition}
\theoremstyle{definition}

\theoremstyle{definition}
\newtheorem{observation}{Observation}

\usepackage{enumitem}
\DeclareMathOperator*{\minimize}{minimize}

\DeclareMathOperator{\argmax}{argmax}

\usepackage{mathtools}

\DeclarePairedDelimiter\floor{\lfloor}{\rfloor}
\newcommand{\cut}{\textbf{cut}}
\newcommand{\st}{$s$-$t$}

\newcommand{\V}{\mathcal{V}}
\newcommand{\E}{\mathcal{E}}

\newcommand{\vw}{\textbf{w}}

\newcommand{\vf}{\textbf{f}}

\newcommand{\vhw}{\hat{\textbf{w}}}
\newcommand{\vhf}{\hat{\textbf{f}}}

\usepackage{tikz}

\usepackage{authblk}

\title{On the tractability and approximability of non-submodular cardinality-based $s$-$t$ cut problems in hypergraphs\thanks{Both authors are supported by the Army Research Office (ARO award \#W911NF‐24-1-0156). We thank Jon Kleinberg and Magnus Wahlstr{\"o}m for several helpful conversations.} }

\author{Vedangi Bengali\thanks{Department of Computer Science and Engineering, Texas A\&M University; \texttt{vedangibengali@tamu.edu}.} \hspace{1cm}
Nate Veldt\thanks{Department of Computer Science and Engineering, Texas A\&M University; \texttt{nveldt@tamu.edu}.}}
\date{}

\begin{document}

	\maketitle
	\begin{abstract}
    A minimum $s$-$t$ cut in a hypergraph is a bipartition of vertices that separates two nodes $s$ and $t$ while minimizing a hypergraph cut function. The \emph{cardinality-based} hypergraph cut function assigns a cut penalty to each hyperedge based on the number of nodes in the hyperedge that are on each side of the split. Previous work has shown that when hyperedge cut penalties are submodular, this problem can be reduced to a graph $s$-$t$ cut problem and hence solved in polynomial time. NP-hardness results are also known for a certain class of non-submodular penalties, though the complexity remained open in many parameter regimes. In this paper we highlight and leverage a connection to Valued Constraint Satisfaction Problems to show that the problem is NP-hard for all non-submodular hyperedge cut penalty, except for one trivial case where a 0-cost solution is always possible. We then turn our attention to approximation strategies and approximation hardness results in the non-submodular case. We design a strategy for projecting non-submodular penalties to the submodular region, which we prove gives the optimal approximation among all such projection strategies. We also show that alternative approaches are unlikely to provide improved guarantees, by showing it is UGC-hard to obtain a better approximation in the simplest setting where all hyperedges have exactly 4 nodes.
    \end{abstract}

        \section{Introduction}

A cut in a graph is a set of edges whose removal partitions the nodes into disconnected components. Finding small graph cuts is a common algorithmic primitive for clustering and partitioning applications such as image segmentation, community detection in social networks, and workload partitioning tasks in parallel computing. In the past several years there has been an increasing interest in solving cut problems over \textit{hypergraphs}, where nodes are organized into multiway relationships called hyperedges~\cite{ccatalyurek2023more, veldt2022hypergraph,zhu2022hypergraph,chekuri2018minimum,panli2017inhomogeneous,panli_submodular}. While edges in a graph model pairwise relationships, hyperedges can directly model multiway relationships such as group social interactions, groups of biological samples with similar gene expression patterns, chemical interactions involving multiple reagents, or groups of interdependent computational tasks in parallel computing applications. 

Given a bipartition of nodes, the standard hypergraph cut function simply counts the number of hyperedges that span both partitions~\cite{lawler1973}. Although this is a straightforward generalization of the graph cut function, it does not capture the fact that there are multiple different ways to split the nodes of a hyperedge into two clusters, each of which may be more or less desirable depending on the application. This had led to recent generalized hypergraph cut functions that assign different cut penalties depending on how the hyperedge is split, with generalized cut penalties captured by a \emph{splitting function} defined for each hyperedge~\cite{veldt2022hypergraph,panli_submodular,panli2017inhomogeneous,zhu2022hypergraph,fountoulakis2021local,chen2023submodular}. Many applications focus specifically on hyperedge cut penalties that are cardinality-based, meaning the penalties depend only on the number of nodes of a hyperedge that are on each side of a split. Veldt et al.~\cite{veldt2022hypergraph} provided a systematic study of the hypergraph $s$-$t$ cut problem for cardinality-based cut functions, showing that this problem can be reduced to an $s$-$t$ cut problem in a directed weighted graph if and only if all hyperedge splitting functions are submodular. These fundamental primitives for hypergraph $s$-$t$ cuts have since been used as subroutines for other hypergraph analysis techniques, such as localized clustering and semi-supervised learning algorithms on large hypergraphs~\cite{liu2021strongly,veldt2020minimizing}, new approaches for dense subhypergraph discovery~\cite{huang2024densest}, and faster algorithms for decomposable submodular function minimization~\cite{veldt2021approximate}. Specific applications of these tools include speeding up algorithms for benchmark image segmentation tasks~\cite{veldt2021approximate,veldt2023augmented}, improved product classification in a large retail product hypergraph~\cite{veldt2020minimizing}, finding related posts in online Q\&A forums~\cite{veldt2020minimizing}, and detecting related restaurants in a large Yelp hypergraph~\cite{liu2021strongly}.

In addition to their algorithmic techniques for submodular cut penalties, Veldt et al.~\cite{veldt2022hypergraph} proved that the cardinality-based $s$-$t$ cut problem is NP-hard for certain specific non-submodular splitting functions. They also highlighted a trivial setting where penalties are not submodular but an optimal (zero-cost) solution can be found by placing one terminal node ($s$ or $t$) in a cluster by itself. The latter rules out the possibility of showing that the cardinality-based $s$-$t$ cut problem is poly-time solvable if and only if cut penalties are submodular. The complexity of the problem remained unknown for a large class of cut penalty choices. Settling these complexity results, even if only for the case of 4-uniform hypergraphs, was later included in a list of open problems in applied combinatorics~\cite{aksoy2023seven}. Very recently, Adriaens et al.~\cite{adriaens2024improved} closed the gap by proving that the problem is NP-hard for all (non-trivial) cut penalties outside the submodular region, subject to the condition that the penalties are polynomially bounded in terms of the size of the hypergraph. A natural question is whether we can establish hardness results that are independent of this assumption about polynomially-bounded cut penalties. Another direction is to develop approximation techniques and refined hardness-of-approximation results for non-submodular cut penalties. 

\paragraph{The present work: finalized complexity results and optimal approximation techniques.}
In our work, we begin by establishing the relationship between generalized hypergraph cut problems and an earlier notion of Valued Constraint Satisfaction Problems (VCSPs) from the theoretical computer science literature~\cite{cohen2006algebraic,cohen2011algebraic}. We prove an equivalence result between the cardinality-based $s$-$t$ cut problem and a specific class of boolean VCSPs. By translating existing complexity dichotomy results for boolean VCSPs, we prove that the cardinality-based hypergraph $s$-$t$ cut problem is NP-hard for every choice of non-submodular cut penalties, aside from the trivial setting where a zero-cost solution is possible. Our approach holds for all finite splitting penalties and thus does not require any assumption about polynomially-bounded penalties.

Next we turn our to approximation algorithms and approximation hardness results for cardinality-based hypergraph $s$-$t$ problems outside the submodular region. For 4-uniform hypergraphs, the complexity of the problem depends on a single cut penalty $w_2$ (the penalty for splitting a hyperedge evenly), and there is a simple strategy for converting a non-submodular splitting function to the closest submodular function, with known approximation factors~\cite{veldt2022hypergraph}. For larger hyperedges, there can be far more cut penalties (since there are more ways to split a hyperedge), and finding the best way to project a non-submodular function to a submodular function is more nuanced. We provide a simple strategy for this projection step, by viewing it as a piecewise linear function approximation problem. We prove our strategy provides the optimal approximation factor among all methods that are based on replacing a non-submodular function with a submodular function. We complement these approximation techniques with approximation hardness results, specifically for the simplest setting where the hypergraph is $4$-uniform. By considering the Basic Linear Programming relaxation for the underlying VCSP~\cite{ene2013local}, we prove that our approximations based on projecting to the submodular region are the best possible assuming the Unique Games Conjecture.

	\section{Technical Preliminaries}
\label{sec:prelims}
We begin by reviewing formal definitions for generalized hypergraph cut problems, along with needed technical background on Valued Constraint Satisfaction Problems.

\subsection{Generalized hypergraph minimum \st{} cuts}
Consider a hypergraph $\mathcal{H}=(\mathcal{V}, \mathcal{E})$, where each hyperedge $e \in \E$ is a set of $|e|\ge 2$ nodes. Given a set of nodes $S$, a hyperedge is cut when its nodes are split between sets $S$ and $\bar S = V\backslash S$.  
A straightforward extension of the graph cut function to hypergraphs is given by the commonly studied \textit{all-or-nothing} cut function, which simply counts the number of hyperedges crossing a bipartition (or the sum of scalar weights of hyperedges if the hypergraph is weighted)~\cite{lawler1973}. Thus, every way of splitting the nodes of a single hyperedge leads to the same penalty for cutting that hyperedge.

A more general approach involves \textit{hyperedge splitting functions}, which assign a non-negative penalty for each of the $2^e$ potential ways a hyperedge $e$ can be split~\cite{veldt2022hypergraph,li2018submodular,panli2017inhomogeneous}. 
Formally, we define a splitting function $\vw_e:2^e\rightarrow \mathbb{R}$ for each hyperedge $e\in \E$ that satisfies the following properties:
\begin{align*}
    &\vw_e(A) \ge 0 \quad \quad \forall A \subseteq e\\
    &\vw_e(A) = \vw_e(e\setminus A) \quad \quad \forall A \subseteq e\\
    &\vw_e(e) = \vw_e(\emptyset) = 0.
\end{align*}
For a set $S \subseteq \V$ and $\{s,t\} \subseteq \V$ as source and sink nodes respectively, the generalized hypergraph minimum \st{} cut problem is then defined to be:
\begin{align}
\label{eq:hyperstcut}
\begin{split}
    &\minimize \quad \textbf{cut}_{\mathcal{H}}(S) = \sum_{e \in \partial S}\vw_e(e \cap S) = \sum_{e \in \E}\vw_e(e \cap S) \quad \text{subject to } \; s\in S, t\in \bar S
\end{split}
\end{align}
where $\partial S = \{ e \in \E : e\cap S \neq \emptyset, e\cap \bar{S} \neq \emptyset\}$ is the set of cut hyperedges. A splitting function is \emph{submodular} if for every $A,B \in 2^e$ it satisfies:
\begin{align}
\label{eq:submodular}
    \vw_e(A \cap B) + \vw_e(A \cup B) \leq \vw_e(A) + \vw_e(B).
\end{align}
When all splitting functions are submodular, the cut function $\cut_\mathcal{H}$ is a sum of submodular functions and hence submodular. The minimum $s$-$t$ cut problem is then polynomial-time solvable, as it is a special case of submodular function minimization.

\subsection{Cardinality-based minimum \st{} cuts}
\textit{Cardinality-based} functions are defined based only on the number of nodes on each side of the partition. Formally, these functions satisfy the additional condition: 
\begin{align}
\label{eq:card}
    \vw_e(A) = \vw_e(B) \quad \forall A,B \in 2^e \text{ where } |A|=|B|.
\end{align}
For a hyperedge of size $r = |e|$, $\textbf{w}_e$ can be completely characterized by $q=\floor{r/2}$ splitting penalties denoted as $w_i$ for $i \in \{0,1,... , q\}$, where $\textbf{w}_e(A) = w_i$ is the penalty for every $A \subseteq e$ such that $\min \{|A|,|e\backslash A|\} = i$. Observe that $w_0 = 0$ always. We also often refer to these as \textit{splitting parameters}
as they can be viewed as parameters defining a class of hypergraph $s$-$t$ cut problems.

\paragraph{The $r$-\cbcut{} problem.}
In practice, hyperedges can be of different sizes and there may be reasons to consider associating different splitting functions to different hyperedges. However, for the purpose of understanding fundamental tractability results, we restrict our attention to $r$-uniform hypergraphs where all hyperedges have the same cardinality-based splitting function. Tractability and hardness results for other cardinality-based hypergraph \st{} cut problems (where hyperedges can have different sizes and splitting functions) can be easily determined by extending results for the uniform case.

Formally, let $\mathcal{H} = (\V, \mathcal{E})$ be an $r$-uniform hypergraph and let $\{w_1, w_2, \hdots, w_q\}$ be a set of non-negative splitting penalties where $q = \floor{r/2}$. The $r$-\cbcut{} problem is given by
\begin{align}
 \label{eq:cbcut}
    \minimize \quad \textbf{cut}_{\mathcal{H}}(S) = \sum_{i=1}^{q} w_i \cdot |\partial S_i| \quad \text{ subject to } s \in S \text{ and } t \in \bar{S}
\end{align}
where $\partial S_i = \{e \in \E:|S\cap e| \in \{i,r-i\}\}$ is the set of hyperedges split by the set $S$ with $i$ nodes on the smaller side of the cut.  When we need to explicitly specify the splitting penalties $\{w_1, w_2, \hdots, w_q\}$, we will denote this problem by $\text{CBcut}(r,\{w_1, w_2, \hdots, w_q\})$. We say that $r$-CB-cut is \emph{tractable} for the splitting parameters $\{w_1, w_2, \hdots, w_q\}$ if every instance of $\text{CBcut}(r,\{w_1, w_2, \hdots, w_q\})$ can be solved in polynomial time. We say that it is NP-hard for these splitting parameters if there exists an NP-hard problem that can be reduced to $\text{CBcut}(r,\{w_1, w_2, \hdots, w_q\})$ in polynomial time. 

For some of our results it will be convenient to consider weighted variants of the problem. For the first variant, we assume that each hyperedge $e \in \E$ is associated with a scalar rational weight $\omega_e \geq 0$, and we scale the cut penalty at this edge by this weight. Formally, the \textsc{weighted} $r$-\cbcut{} problem is then defined to be
\begin{align}
 \label{eq:weightedcard}
    \minimize \quad \textbf{cut}_{\mathcal{H}}(S) = \sum_{i=1}^{q} w_i \cdot W(\partial S_i) \quad \text{ subject to } s \in S \text{ and } t \in \bar{S}
\end{align}
where $W(\partial S_i) = \sum_{e \in \partial S_i} \omega_e$. Another slight variant is to allow a mixture of size-$r$ hyperedges and size-2 hyperedges (i.e., standard graph edges), both of which can be weighted. Formally we are given a hypergraph $\mathcal{H} = (\V, \E \cup E)$ where $\E$ is a set of (scalar-weighted) hyperedges of size-$r$, and $E$ is a set of weighted edges. The $r$-\cbcut{} \textsc{with edges} problem is given by
\begin{align}
 \label{eq:cardwithedges}
    \minimize \quad \textbf{cut}_{\mathcal{H}}(S) = \cut_E(S) + \sum_{i=1}^{q} w_i \cdot W(\partial S_i) \quad \text{ subject to } s \in S \text{ and } t \in \bar{S}
\end{align}
where $\partial S_i$ denotes the size-$r$ hyperedges with $i$ nodes on the small side of the split, and $\textbf{cut}_E(S)$ is the standard graph cut function for the graph defined by edges $E$. These objectives are all closely related. Objective~\eqref{eq:cbcut} is a special case of Objective~\eqref{eq:weightedcard}. An approximation-preserving reduction in the other direction is possible if we scale weights to be integers and replace weights with multiple copies of the same hyperedge. For this reduction we just need scalar weights to be polynomial in terms of the hypergraph size. We also have the following observation.
\begin{observation}
    Objectives~\eqref{eq:weightedcard} and~\eqref{eq:cardwithedges} are equivalent with respect to approximations.
\end{observation}
To see why, note that \textsc{Weighted} $r$-\cbcut{} is a special case of $r$-\cbcut{} \textsc{with edges} where $E = \emptyset$. We can reduce the latter to the former in an approximation-preserving way by replacing each $(x,y) \in E$ with a size-$r$ hyperedge $(x,y, a_1, \hdots, a_{r-2})$, where $(a_1, \hdots, a_{r-2})$ are new nodes that only show up in this hyperedge. We see therefore that proving a result for one of these objectives applies the same result for the other (modulo certain assumptions about polynomially-bounded weights in some cases).

\paragraph{Prior tractability and hardness results.}
Veldt et al.~\cite{veldt2022hypergraph} proved that a cardinality-based splitting function for an $r$-node hyperedge is submodular if and only if its splitting penalties satisfy
\begin{equation}
\label{eq:submodular}
\begin{split}
        2w_1 &\geq w_2\\
        2w_j &\geq w_{j-1}+w_{j+1} \text{ for } j= 2,\hdots,q-1\\
        w_q &\geq w_{q-1} \ge \hdots \ge w_2 \ge w_1 \ge 0.
\end{split}
\end{equation}
Thus, if $\{w_1, w_2, \hdots, w_q\}$ satisfies these conditions, \cbw{} is polynomial-time solvable. The problem is also (trivially) polynomial-time solvable if $w_1 = 0$, since in this case one can set $S = \{s\}$ and the resulting cut penalty is zero. This holds independent of the values for $\{w_2, w_3, \hdots, w_q\}$, including values for which the overall splitting function is non-submodular. This problem is called \textsc{Degenerate} \cbcut{}. If $w_1 > 0$, we can scale all penalties without loss of generality so that $w_1 = 1$; we will typically assume this scaling for all non-degenerate cases we consider throughout the manuscript. Veldt et al.~\cite{veldt2022hypergraph} proved that $r$-\cbcut{} is NP-hard, via reduction from maximum cut, for splitting penalties in the non-submodular regime that satisfy the following two constraints:
    \begin{align*}
        &w_1 > w_j > 0 \quad \text{where $j \in {2,3,\ldots,\floor {r/2}}$}\\
        &w_i \ge w_j \quad \forall i \neq 1.
    \end{align*}
In recent independent and concurrent work, Adriaens et al.~\cite{adriaens2024improved} showed improved reductions from maximum cut to prove that \cbw{} is NP-hard for all non-submodular penalties, i.e., all choices of $\{w_1 > 0, w_2, \hdots, w_q\}$ that violate any of the constraints in Eq.~\eqref{eq:submodular}. For their reduction, splitting penalties must be polynomially bounded in terms of the hypergraph size.

\subsection{Valued Constraint Satisfaction Problems}
\textit{Valued Constraint Satisfaction Problems} (VCSPs) provide a general framework to model and solve optimization problems over a language involving variables, constraints, and value assignments. Cohen et al.~\cite{cohen2004complete} studied the complexity of VCSPs defined over Boolean variables. We will model the $r$-\cbcut{} problem as a special type of Boolean VCSP and then translate existing complexity results for the latter problem to hypergraph cut problems. We review the notation and definitions established by Cohen et al.~\cite{cohen2004complete}.\footnote{These authors gave generalized definitions for VSCPs that also apply to the non-Boolean case, but we restrict to Boolean variables since this suffices for our purposes.}

An instance of Boolean VCSP is given by a tuple $\mathcal{P} = \langle V,C,\mathcal{X} \rangle$ where
\begin{itemize}
    \item $V = \{v_1, v_2, \hdots, v_n\}$ is a finite set of Boolean variables.
    \item $\mathcal{X}\subseteq \mathbb R^+$ is a set of possible costs.
    \item $C$ is a set of constraints, each defined by a pair $c = \langle \sigma,\phi \rangle$, where the \emph{scope} $\sigma \subseteq V$ defines a set of variables the constraint applies to, and $\phi: \{0,1\}^{|\sigma|} \rightarrow \mathcal{X} $ is a \emph{cost} function that maps every possible assignment of these variables to a cost in $\mathcal{X}$.
\end{itemize}
If $\phi \colon \{0,1\}^m \rightarrow \mathcal{X}$, the value $m$ is called the \emph{arity} of $\phi$, and it only applies to scopes of size $m$. For a constraint $c = \langle \sigma, \phi \rangle$ where $|\sigma| = m$, we denote this scope by $\sigma = \{v_{\sigma_1}, v_{\sigma_2}, \hdots, v_{\sigma_{m}} \} \subseteq V$. For example, if $\sigma = \{v_2, v_4, v_{10}\}$, then $|\sigma| = 3$, $\sigma_1 = 2$, $\sigma_2 = 4$, and $\sigma_{3} = 10$. The goal is to find an assignment $a \colon V \rightarrow \{0,1\}^n$ of variables to Boolean values to solve the following problem:
\begin{equation}
\label{eq:vcsp}
    \minimize_{a \colon V \rightarrow \{0,1\}^n} \quad \text{Cost}_{\mathcal{P}}(a) = \sum_{\langle \sigma,\phi\rangle \in C} \phi(a(v_{\sigma_1}),a(v_{\sigma_2}),\hdots,a(v_{\sigma_{|\sigma|}}))
\end{equation}
Prior work has focused on proving complexity results for this objective under different assumptions about the cost functions $\phi$. Formally, let $\Gamma$ represent a collection of cost functions. A \textit{Valued Boolean Constraint Language} VCSP($\Gamma$) is then a tuple $\langle V, \mathcal{X}, \Gamma, C \rangle$ where cost functions in the constraint set $C$ come from the collection $\Gamma$. The constraint language is called \emph{tractable} when all instances within VCSP($\Gamma$) can be solved in polynomial time. It is {NP-hard} (as an entire language) if an existing NP-hard problem has a polynomial-time reduction to VCSP($\Gamma$).

\paragraph{Relation to hypergraph $s$-$t$ cuts.} There is a close connection between the VCSP objective in~\eqref{eq:vcsp} and the generalized hypergraph $s$-$t$ cut problem in~\eqref{eq:hyperstcut}. The Boolean variables $V$ can be thought of as (non-terminal) nodes $\V - \{s,t\}$ in a hypergraph $\mathcal{H} = (\V, \E)$. Assigning a variable to 1 corresponds to placing a node on the $s$-side of a cut, while assigning to 0 means placing the node on the $t$-side. Each scope $\sigma$ of a constraint in the VCSP problem corresponds to the nodes in some hyperedge $e$, and the cost function $\phi$ corresponds to the splitting function $\vw_e$, which gives a penalty for each way of splitting the nodes. Proving tractability results for a cost function collection $\Gamma$ then corresponds to proving tractability results for a collection of hypergraph $s$-$t$ cut problems defined by a class of splitting functions (in our case, cardinality-based splitting functions for a specific choice of splitting penalties $\{w_1, w_2, \hdots, w_q\}$). This relationship allows us to translate existing tractability results from the VCSP literature to hypergraph cut problems, though a few additional details are needed to ensure the relationship is formalized correctly (see Section~\ref{sec:nphard}).

\paragraph{Complexity dichotomy results for Boolean VCSPs.}
Cohen et al.~\cite{cohen2004complete} proved complete complexity dichotomy results for Boolean VCSPs, which rely on proving certain inequalities involving the notion of a \emph{multimorphism}. To summarize these results, we must first extend a cost function $\phi \colon \{0,1\}^m \rightarrow \mathcal{X}$ of arity $m$ so that it can be applied to $m$ \emph{tuples} of Boolean values, rather than just $m$ Boolean values. Formally, if $\{t_1, t_2, \hdots, t_m\} \subseteq \{0,1\}^k$ is a set of $m$ Boolean vectors of size $k$, where $t_j[i]$ is the $i$th entry of the $j$th vector, we define
\begin{equation}
    \phi(t_1,t_2,...,t_m) = \sum_{i=1}^k \phi(t_1[i],t_2[i],...t_m[i]).
\end{equation}
In other words, evaluating $\phi$ on these tuples corresponds to evaluating them $k$ times (one for each position in the vectors), and then summing the results. A function $F \colon \{0,1\}^k \rightarrow \{0,1\}^k$ is defined to be a \emph{multimorphism} of $\phi$ if the following inequality holds:
\begin{equation}
\label{eq:multimorphism}
\phi(F(t_1),F(t_2),...,F(t_m)) \le \phi(t_1,t_2,...,t_m).
\end{equation}
If $\Gamma$ is a collection of cost functions, then we say $F$ is a multimorphism of $\Gamma$ if $F$ is a multimorphism of every $\phi \in \Gamma$. Cohen et al.~\cite{cohen2004complete} proved the following complexity dichotomy result for languages with finite-valued cost functions.
\begin{theorem}[Corollary 2~\cite{cohen2004complete}]
\label{thm:vcsp-trac}
If $\mathcal{X}$ involves only finite values, VCSP($\Gamma$) is tractable if $\Gamma$ has any of the following multimorphisms $\langle \textbf{0} \rangle , \langle \textbf{1} \rangle \text{ or } \langle \min,\max \rangle$; otherwise VCSP($\Gamma$) is NP-hard. 
\end{theorem}
In this theorem, $\langle \textbf{0} \rangle$ represents the function that maps everything to 0, and $\langle \textbf{1} \rangle$ similarly maps everything to 1. The function $\langle \min,\max \rangle \colon \{0,1\}^2 \rightarrow \{0,1\}^2$ is defined by:
\begin{equation}
    \langle \min,\max \rangle (x_1, x_2) = (\min\{x_1, x_2\}, \max \{x_1, x_2\}).
\end{equation}
It is known that this function is a multimorphism of a cost function $\phi$ if and only if $\phi$ is submodular~\cite{cohen2003soft,cohen2004complete}. We provide a proof to give an intuition for multimorphisms.
\begin{lemma}
\label{lem:submodular}
    The cost function $\phi \colon \{0,1 \}^m \rightarrow \mathbb{R}^+$ is submodular if and only if it has $\langle \min,\max \rangle$ as a multimorphism.
\end{lemma}
\begin{proof}
    The definition of submodularity given for set functions in Eq.~\eqref{eq:submodular} can be translated easily to a property of a Boolean function $\phi$. Formally, consider two $m$-tuples of binary values $\textbf{a} = (a_1, a_2, \hdots, a_m)$ and $\textbf{b} = (b_1, b_2, \hdots, b_m)$, which we can think of as indicator vectors for some sets $A$ and $B$. The indicator vectors for sets $A \cup B$ and $A \cap B$ are given by:
   \begin{align}
       \textbf{a}\cap\textbf{b} &= (\min \{a_1, b_1\}, \min \{a_2, b_2\}, \hdots ,\min \{a_m, b_m\})\\
       \textbf{a}\cup\textbf{b} &= (\max \{a_1, b_1\}, \max \{a_2, b_2\}, \hdots ,\max \{a_m, b_m\}).
   \end{align}
  By definition, $\phi$ is submodular if it satisfies the constraint
   \begin{equation} 
   \label{eq:booleansubmodularity}
   \phi(\textbf{a}\cap\textbf{b}) + \phi(\textbf{a}\cup\textbf{b}) \leq \phi(\textbf{a}) + \phi(\textbf{b})
   \end{equation}
   for an arbitrary pair of Boolean vectors \textbf{a} and \textbf{b}. 
   Define now a set of $m$ 2-tuples $\{t_1, t_2, \hdots, t_m\}$ by stacking \textbf{a} and \textbf{b}, so that $t_i = (\textbf{a}[i], \textbf{b}[i])$. The definition of a multimorphism in Eq.~\eqref{eq:multimorphism} applied to $\langle \min, \max \rangle$ exactly corresponds to the inequality defining submodularity in Eq.~\eqref{eq:booleansubmodularity}.
\end{proof}

	\section{NP-hardness for Non-submodular \cbcut{}}
\label{sec:nphard}
We now prove NP-hardness for all cardinality-based $s$-$t$ cut problems with non-submodular parameters and $w_1 = 1$. Recall that when $w_1 = 0$ the problem is trivial. As a warm-up, and to set the stage for some of our approximation hardness results in Section~\ref{sec:approxhard}, we show NP-hardness for $4$-\cbcut{} whenever $w_2 \notin [1,2]$ via direct reduction from maximum cut (\textsc{MaxCut}). We then leverage the connection to VCSPs to prove hardness for arbitrary-sized hyperedges. The main result of this section is to confirm that submodularity not only coincides exactly with graph-reducibility (as shown by Veldt et al.~\cite{veldt2022hypergraph}), but also coincides with tractability.
\begin{theorem}
\label{thm:nphard}
    If $w_1 = 1$, \cbw{} is tractable if and only if the splitting function is submodular. 
\end{theorem}
In other words, the only non-submodular case where the problem is tractable is \textsc{Degenerate} \cbcut{} (where $w_1 = 0$).

\subsection{Direct reductions for $4$-\cbcut{}}
\label{sec:4cbcut}
\begin{figure}
    \centering
    \begin{minipage}{0.48\textwidth}
        \centering
        \includegraphics[width=1\textwidth]{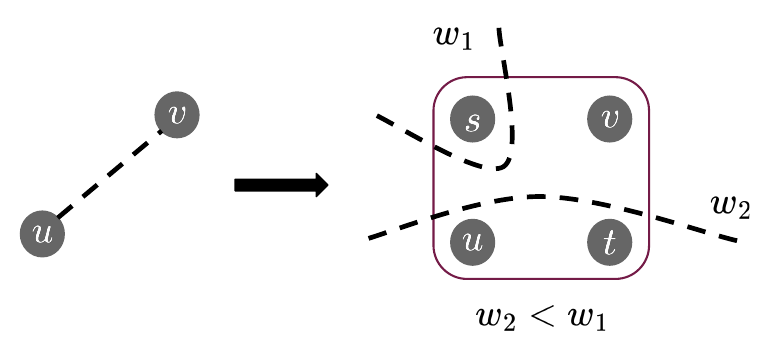} 
    \end{minipage}\hfill
    \begin{minipage}{0.48\textwidth}
        \centering
        \includegraphics[width=1\textwidth]{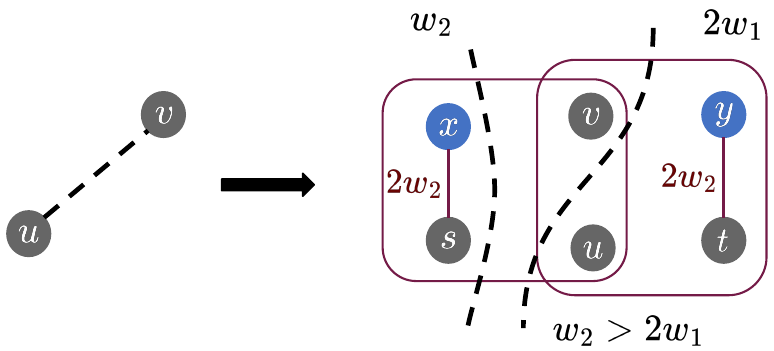} 
    \end{minipage}
    \caption{Gadgets for reducing \textsc{MaxCut} to 4-\cbcut{} when $w_2 < 1$ (left, first shown in~\cite{veldt2022hypergraph}) and $w_2 > 2$ (right).}
    \label{fig:4cbcut}
\end{figure}
For $4$-\cbcut{}, there are only two ways to penalize a split hyperedge. The penalty is $w_1 = 1$ for a (1,3)-split, and is $w_2$ for a (2,2)-split. The submodular region is given by $w_2 \in [1,2]$. Let $G = (V,E)$ be an undirected unweighted graph representing an instance of \textsc{MaxCut}, a well-known NP-hard problem. The goal of \textsc{MaxCut} is to separate $V$ into two groups in order to maximize the number of edges crossing between groups. When $w_2 < 1$, Veldt et al.~\cite{veldt2022hypergraph} showed a simple reduction from \textsc{MaxCut} to $4$-\cbcut{} by introducing source and sink nodes $s$ and $t$ and replacing each edge $(u,v)$ with a hyperedge $e_{uv} = (u,v,s,t)$ (see Figure~\ref{fig:4cbcut}, left panel). Every hyperedge must be cut since it involves both $s$ and $t$. If $u$ and $v$ are separated, the penalty at $e_{uv}$ is 1, otherwise it is $w_2 < 1$. Since it is cheaper to split hyperedges evenly, the $4$-\cbcut{} problem is equivalent to cutting as many hyperedges evenly as possible, which is the same as maximizing the number of cut edges in the \textsc{MaxCut} instance.

We can obtain a similar reduction for $w_2 > 2$ to $4$-\cbcut{} \textsc{with edges} by replacing each edge in $G$ with a slightly more involved gadget. Replace $(u,v) \in E$ with two hyperedges $(u,v,s,x)$ and $(u,v,t,y)$, and two edges $(s,x)$ and $(t,y)$ of weight $2w_2$, which is large enough that they would never be cut by an optimal solution (see Figure~\ref{fig:4cbcut}, right panel). If $u$ and $v$ are together (either on the $s$-side or $t$-side), the penalty would be $w_2$, otherwise the penalty would be $2w_1 = 2 < w_2$. Again, it is cheaper to separate $u$ and $v$, so the hypergraph $s$-$t$ cut problem becomes equivalent to \textsc{MaxCut} at optimality. 

Our new gadget for $w_2 > 2$ is closely related to the gadget developed independently and concurrently by Adriaens et al.~\cite{adriaens2024improved}. Both our reduction from \textsc{MaxCut} and their reduction from \textsc{MaxCut} work for (unweighted) $4$-\cbcut{} only if $w_2$ is polynomially bounded in the hypergraph size. Adriaens et al.~\cite{adriaens2024improved} also showed a reduction from 3SAT to prove NP-hardness for $4$-\cbcut{} when $w_2 > 2$ without needing this polynomially-bounded assumption. We will next show a reduction from a certain class of VCSPs that proves NP-hardness for arbitrary-sized hyperedges, that do not require costs to be polynomially bounded.

\subsection{Hardness for arbitrary hyperedge sizes via VCSPs}
We now introduce a specific Valued Boolean Constraint Language that we will show exactly corresponds to $r$-\cbcut{}. Fix a value of $r$ and parameters $\{w_1, w_2, \hdots, w_q\}$. Define $\Gamma_{r,w_1,\hdots,w_q}$ to be the constraint language that includes exactly four cost functions $\{\phi_r,\phi_s,\phi_t,\phi_{st}\}$.
%
%
We first define the cost function $\phi_r \colon \{0,1\}^r \rightarrow \{w_1, w_2, \hdots, w_q\}$ by 
\begin{align}
    \phi_r(x_1,x_2,\hdots, x_r) = 
    \begin{cases}
        w_j & \text{if } \sum_{i = 1}^r x_i \in \{0,1, \hdots, q\} \\
        w_{r-j} & \text{if } \sum_{i = 1}^r x_i \in \{q+1,q+2, \hdots, r \}.
    \end{cases}
 \end{align}
 This corresponds to a cardinality-based splitting function with parameters $\{w_1, w_2, \hdots, w_q\}$, applied to a hyperedge with $r$ non-terminal nodes. The other three cost functions $\{\phi_s, \phi_t, \phi_{st}\}$ correspond to the same splitting function applied to hyperedges containing $s$ but not $t$, containing $t$ but not $s$, or containing both $s$ and $t$, respectively. We think of $x_i = 1$ as assigning a variable to the $s$-side of the cut and $x_i = 0$ as assigning a variable to the $t$-side. These three cost functions are therefore defined by applying $\phi_r$ with one or two input variables fixed to 0 or 1:
\begin{align*}
\phi_s(x_1, x_2, \hdots, x_{r-1}) &= \phi_r(1,x_1, x_2, \hdots, x_{r-1}) \\
\phi_t(x_1, x_2, \hdots, x_{r-1}) &= \phi_r(x_1, x_2, \hdots, x_{r-1}, 0) \\
\phi_{st}(x_1, x_2, \hdots, x_{r-2}) &= \phi_r(1,x_1, x_2, \hdots, x_{r-2}, 0)
\end{align*}

\begin{lemma}
\label{lem:equivalence}
VCSP($\Gamma_{r,w_1,\hdots,w_q}$) is tractable if and only if \cbw{} is tractable.
\end{lemma}
\begin{proof}
    The reduction in both directions is straightforward; we show one direction for clarity. Let $\mathcal{H} = (\V,\E)$ denote an $r$-uniform hypergraph with node set $\V = \{b_0 = s, b_1, b_2, \hdots, b_n, b_{n+1} = t\}$. For a hyperedge $e$, let $e(i)$ denote the index of the $i$th node in $e$, so that the hyperedge can be expressed as $e = (b_{e(1)}, b_{e(2)}, \hdots, b_{e(r)}) \subseteq \V$. If $e$ contains $s$, we assume nodes in $e$ are ordered so that $b_{e(1)} = s$. If $e$ contains $t$, we assume $b_{e(r)} = t$. The ordering for non-terminal nodes is arbitrary.  

    To reduce this instance of \cbw{} to an instance of VCSP($\Gamma_{r,w_1,\hdots,w_q}$), we introduce a set of $n$ variables $V = \{v_1, v_2, \hdots, v_n\}$. For a hyperedge $e$ not containing terminal nodes, we add a constraint $\langle \langle v_{e(1)}, v_{e(2)}, \hdots, v_{e(r)} \rangle, \phi_r \rangle$. If $e$ contains $s$ but not $t$, we add constraint $\langle \langle v_{e(2)}, \hdots, v_{e(r)} \rangle, \phi_s \rangle$ (since $b_{e(1)} = s$). If it contains $t$ but not $s$, we add constraint $\langle \langle v_{e(1)}, v_{e(2)}, \hdots, v_{e(r-1)} \rangle, \phi_t \rangle$ (since $b_{e(r)} = t$). If it contains both $s$ and $t$, we add constraint $\langle \langle v_{e(2)}, v_{e(3)}, \hdots, v_{e(r-1)} \rangle, \phi_{st} \rangle$. It is straightforward to check that there is a variable assignment with cost $\alpha \geq 0$ for the VCSP instance if and only if there is an $s$-$t$ cut with cut value $\alpha$ in the hypergraph $\mathcal{H}$. The reduction from VCSP($\Gamma_{r,w_1,\hdots,w_q}$) to \cbw{} is similar. 
\end{proof}

\paragraph{Complexity dichotomy results for $r$-\cbcut{}}
At this point we simply need to interpret Theorem~\ref{thm:vcsp-trac} to determine complexity dichotomy results from the constraint language $\Gamma_{r,w_1,\hdots,w_q}$, which in turn gives complexity dichotomy results for \cbw{}. Note that all of the cost functions $\{\phi_r, \phi_s, \phi_t, \phi_{st}\}$ are submodular if and only if $\phi_r$ is submodular, which is true if and only if splitting parameters $\{w_1, w_2, \hdots, w_r\}$ satisfy the submodularity inequalities in Eq.~\eqref{eq:submodular}. We know from Lemma~\ref{lem:submodular} that $\langle \min, \max \rangle$ is a multimorphism of $\phi_r$ if and only if $\phi_r$ is submodular. This corresponds to the known tractable submodular regime for $r$-\cbcut{}. Theorem~\ref{thm:vcsp-trac} tells us that the only other situation where VCSP($\Gamma_{r,w_1,\hdots,w_q}$) is tractable is when $\langle \textbf{0} \rangle$ or $\langle \textbf{1} \rangle$ is a multimorphism of $\Gamma_{r,w_1,\hdots,w_q}$.
Observe that $F = \langle \textbf{0} \rangle$ is a multimorphism of $\Gamma_{r,w_1,\hdots,w_q}$ if and only if all four of the following inequalities hold for all choices of $\{x_1, x_2, \hdots, x_r\}$:
\begin{align*}
     \phi_r(F(x_1),F(x_2),\hdots,F(x_r)) &= \phi_r(0,0,\hdots,0)
      = 0 \le \phi_r(x_1,x_2,\hdots,x_r) \\
     \phi_s(F(x_1),F(x_2),\hdots,F(x_{r-1})) & = \phi_s(0,0,\hdots,0)
      = \phi_r(1, 0,0,\hdots,0) = w_1 \le \phi_s(x_1,x_2,\hdots,x_{r-1})\\
     \phi_t(F(x_1),F(x_2),\hdots,F(x_{r-1})) & = \phi_t(0,0,\hdots,0)
      = \phi_r(0,0,\hdots,0) = 0 \le \phi_t(x_1,x_2,\hdots,x_{r-1}) \\
     \phi_{st}(F(x_1),F(x_2),\hdots,F(x_{r-2})) & = \phi_{st}(0,0,\hdots,0)
      = \phi_r(1, 0,0,\hdots,0) = w_1 \le \phi_{st}(x_1,x_2,\hdots,x_{r-2})
 \end{align*}
The first and third inequalities are always true, but the second and fourth are true for all inputs if and only if $w_1 = 0.$ We can similarly show that $\langle \textbf{1} \rangle$ is a multimorphism of $\Gamma_{r,w_1,\hdots,w_q}$ if and only if $w_1 = 0.$ Thus, aside from the submodular regime, the only tractable case for \cbw{} is \textsc{Degenerate} $r$-\cbcut{}. This proves Theorem~\ref{thm:nphard}.

        \section{Approximating Non-submodular $r$-\cbcut{}}
Having established complexity dichotomy results for $r$-\cbcut{}, we would like to determine the best approximation guarantees we can achieve for the NP-hard non-submodular cases. We will specifically design approximation algorithms that rely on projecting a set of non-submodular splitting penalties (i.e., values $\{w_1 = 1, w_2, \hdots, w_q\}$ that \emph{do not} satisfy inequalities in~\ref{eq:submodular}) to a nearby set of submodular penalties (``nearby'' values $\{\hat{w}_1, \hat{w}_2, \hdots, \hat{w}_q\}$  that do satisfy these inequalities). We can then solve the latter submodular hypergraph $s$-$t$ cut problem to provide an approximation for the original non-submodular problem. 

\textbf{Known approximations for $4$-\cbcut{}.} Veldt et al.~\cite{veldt2022hypergraph} previously applied this approach to obtain simple approximation guarantees for $r$-\cbcut{} for $r \in \{4,5\}$ in the non-submodular region, i.e., when $w_1 = 1$ and $w_2 \notin [1,2]$. Finding a cardinality-based splitting function that is ``closest'' to a non-submodular splitting function is particularly easy in this case since it just involves either decreasing or increasing $w_2$, depending on whether it lies to the right or left of the submodular region $[1,2]$; see Figure~\ref{fig:4cb-projection}. To approximate the problem when $w_2 < 1$, one can compute the solution to the $s$-$t$ cut problem with splitting penalty $\hat{w}_2 = 1$. It is not hard to show this produces a $\frac{1}{w_2}$-approximate solution for the original NP-hard problem. Similarly, if $w_2 > 2$, one can solve the submodular problem where $\hat{w}_2 = 2$ to obtain a $w_2/2$ approximation.

\begin{figure}
    \centering
    \subfigure[$r = 4 \text{ or } 5$]{%
    \centering
        \includegraphics[width=0.35\linewidth]{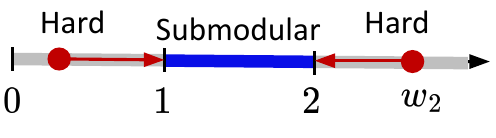}
        \label{fig:4cb-projection}
    }
    \hfill
    \subfigure[$r = 6 \text{ or } 7$]{%
    \centering
        \includegraphics[width=0.35\linewidth]{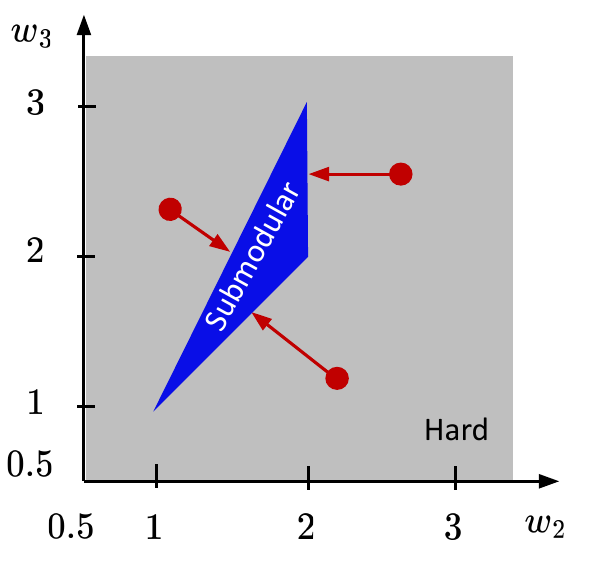} %
        \label{fig:67cb-projection}
    }
    \caption{Submodular (shaded in blue) and NP-hard (shaded in gray) regions of $r$-\cbcut{} for different values of $r$ and  fixed $w_1 = 1$.}
    \label{fig:}
\end{figure}

\textbf{Projecting splitting penalties for $r>5$.}
Projecting non-submodular splitting penalties to the submodular region becomes more nuanced when the hyperedge size is $r > 5$ (see Figure~\ref{fig:67cb-projection}). This typically involves changing multiple splitting penalties (not just $w_2$), and it is not clear at prior what it means to find the ``best'' or ``closest'' set of submodular splitting penalties for a given set of non-submodular splitting penalties. The submodularity inequalities in~\eqref{eq:submodular} define a closed convex set, so a natural idea is to project a given set of non-submodular penalties $\{w_1 = 1, w_2, \hdots, w_q\}$ to this convex set in a way that minimizes a 2-norm distance or the distance in terms of some other norm. There are many well-known methods for projecting a point onto a convex set that could be used for this approach. However, will this provide the best possible approximation guarantee for a non-submodular hypergraph $s$-$t$ cut problem? In order to answer this question, we will first show how to characterize the approximation guarantee for any technique that projects non-submodular penalties to the submodular region. We will then design a new approach based on a piecewise-linear function approximation problem, that is guaranteed to provide the best approximation among all methods that project non-submodular points to the submodular region.
 
\subsection{Approximation bounds for non-submodular $r$-\cbcut{}}
We will associate each $r$-\cbcut{} problem with a vector of splitting penalties $\vw = [w_1,w_2,\hdots,w_q]^{\top} \in \mathbb{R}^q_+$ which we call the ``splitting vector.'' The space $\mathbb{R}_+^q$ represents the universe of all potential $r$-\cbcut{} splitting vectors. Within this space, let $\mathcal{S}_q \subseteq \mathbb{R}_+^q $ represent the subset of vectors that correspond to submodular cardinality-based splitting functions. We will refer to these as $r$-SCB (Submodular and Cardinality-Based) vectors. Vectors associated with non-submodular splitting functions form the set $\{\vw \colon \vw \in \mathbb{R}_+^q\setminus\mathcal{S}_q\}$ of $r$-NCB (Non-submodular Cardinality-Based) vectors. We wish to design a method to project each $r$-NCB vector $\vw \notin \mathcal{S}_q$ to an $r$-SCB vector $\hat{\vw} = [\hat{w}_1,\hat{w}_2,\hdots,\hat{w}_q]^{\top}$, in a way that provides the best approximation guarantee for $r$-\cbcut{}. Let $\cut_\mathcal{H}$ denote the cut function for non-submodular splitting vector $\vw$ and $\hat{\cut}_\mathcal{H}$ be the cut function for $\hat{\vw}$. Formally, our goal is to choose $\hat{\vw} \in \mathcal{S}_q$ so that for every $r$-uniform hypergraph $\mathcal{H} = (\V, \E)$ and every $S \subseteq \V$ we have
\begin{align}
    \cut_{\mathcal{H}}(S) \le \hat{\textbf{cut}}_{\mathcal{H}}(S) \le \rho_r\cdot\cut_{\mathcal{H}}(S) \label{eq:max_approx}
\end{align}
for the smallest possible value of $\rho_r$. Note that satisfying the first inequality $\cut_{\mathcal{H}}(S) \le \hat{\textbf{cut}}_{\mathcal{H}}(S)$ requires us to enforce the constraint $\hat{\vw} \ge \vw$. We can enforce this without loss of generality; if the projected vector had a parameter $\hat{w}_j < w_j$ for some $j \in \{1,2,\hdots,q\}$, we could scale the entire vector $\hat{\vw}$  by a factor of $\vw_j/\hat{\vw}_j$ without affecting optimal solutions or approximation guarantees. For convenience, we will assume throughout this section that $\vw(1) = w_1 = 1$ for the original vector $\vw$ that we wish to project to the submodular region. Note that because of the constraint $\hat{\vw} \geq \vw$, it is possible for $\hat{w}_1 > 1$ to hold. 

The following lemma shows how well we can approximate $r$-\cbcut{} for a splitting vector $\vw$ by instead solving a nearby problem defined by splitting vector $\hat{\vw}$. The result holds for an arbitrary pair of splitting vectors satisfying $\vw \leq \hat{\vw}$. In principle, the idea is to apply this to project $\vw \notin \mathcal{S}_q$ to some nearby point $\hat{\vw} \in \mathcal{S}_q$.
\begin{lemma}
\label{lem:projbound}
    Let $\vw, \hat{\vw} \in \mathbb{R}^q_+$ be a pair of splitting vectors satisfying $\vw \leq \hat{\vw}$. For every $r$-uniform hypergraph $\mathcal{H}=(\V,\E)$ and $S \subseteq \V$, the inequalities in~\eqref{eq:max_approx} are satisfied for
    \begin{align*}
        \rho_r = \max_i \frac{\hat{w}_i}{w_i}  \quad  \text{ for }  i \in \{1,2,3,\hdots,q\}.
    \end{align*}
\end{lemma}
\begin{proof}
    For an arbitrary set $S \subseteq \V$, the cut penalty with respect to $\vw$ is $\textbf{cut}_{\mathcal{H}}(S)$ and the cut value with respect to $\hat{\vw}$ is $\hat{\textbf{cut}}_\mathcal{H}(S)$. Thus:
    \begin{align*}
        \frac{\hat{\textbf{cut}}_\mathcal{H}(S)}{\textbf{cut}_{\mathcal{H}}(S)}
        &= \frac{\hat{w}_1\cdot|\partial S_1| + \hat{w}_2\cdot |\partial S_2| +, \hdots,+ \hat{w}_q\cdot |\partial S_q|}{w_1\cdot |\partial S_1| + w_2\cdot|\partial S_2| +, \hdots,+ w_q\cdot |\partial S_q| } \le \max_{i \in \{1,\hdots, q\}} \left( \frac{\hat{w}_i}{w_i} \right).
    \end{align*}
\end{proof}
Observe that the approximation bound of $\max_i \frac{\hat{w}_i}{w_i}$ is tight. Recall that our goal is to project non-submodular splitting penalties in such a way that we can guarantee a certain approximation factor for all instances of \cbw{}. This means that we need the inequality in~\eqref{eq:max_approx} to hold for every possible $r$-uniform hypergraph $\mathcal{H}$ and node set $S$. In the worst case, it is possible to construct a hypergraph such that the approximation factor in Lemma~\eqref{lem:projbound} is tight. More precisely, if the largest ratio is $\hat{w}_t/w_t$ for some $t \in \{1, 2, \hdots, q \}$, we can find a hypergraph with node $S$ where $\partial S$ only contains cut hyperedges with exactly $t$ nodes on the small side of the cut. This results in an approximation factor of exactly $\hat{w}_t/w_t$.

\subsection{Norm-minimizing projection techniques}
\begin{figure}[t]
    \centering
    \subfigure[Bounds obtained via $\ell_1$-norm projection]{%
        \includegraphics[width=0.45\textwidth]{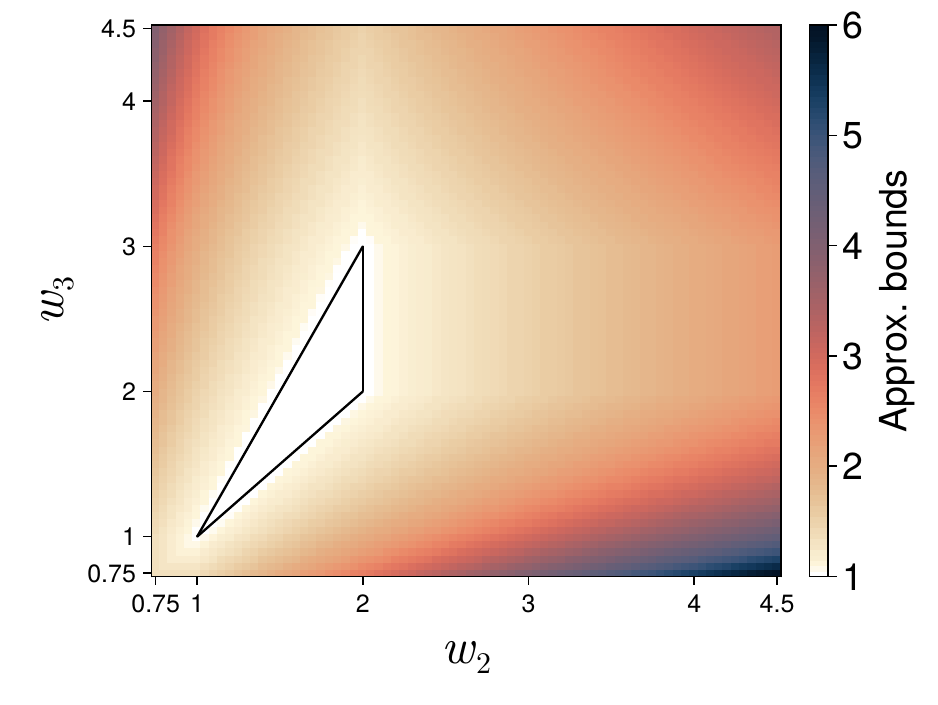}
        \label{fig:heatmap_L1}
    }
    \hfill
    \subfigure[Bounds obtained via $\ell_2$-norm projection]{%
        \includegraphics[width=0.45\textwidth]{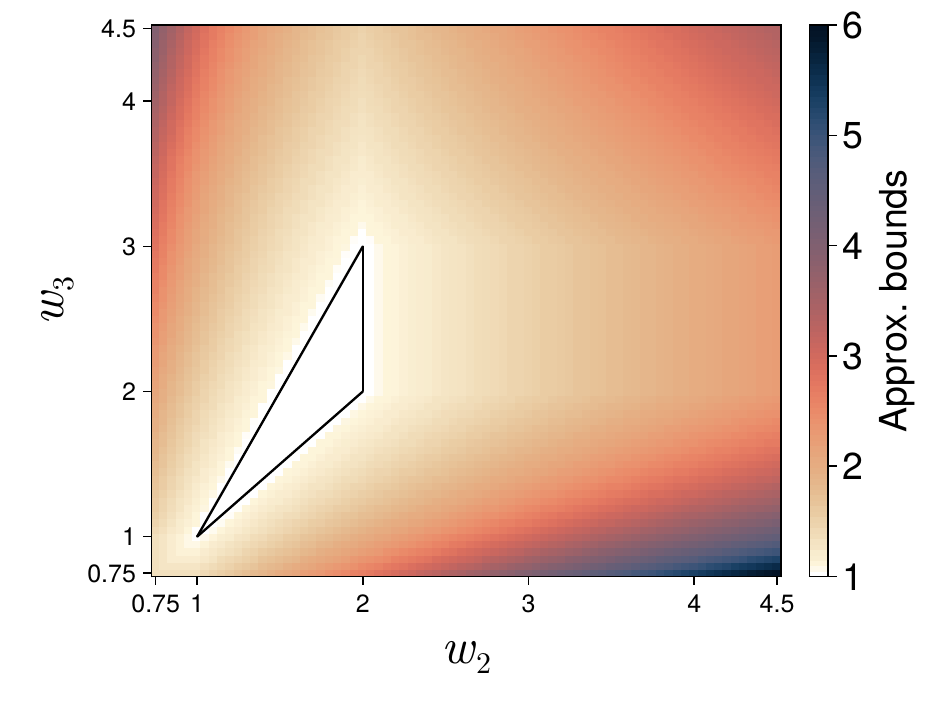} %
        \label{fig:heatmap_L2}
    }
    \hfill
    \subfigure[Bounds obtained via $\ell_\infty$-norm projection]{%
        \includegraphics[width=0.45\textwidth]{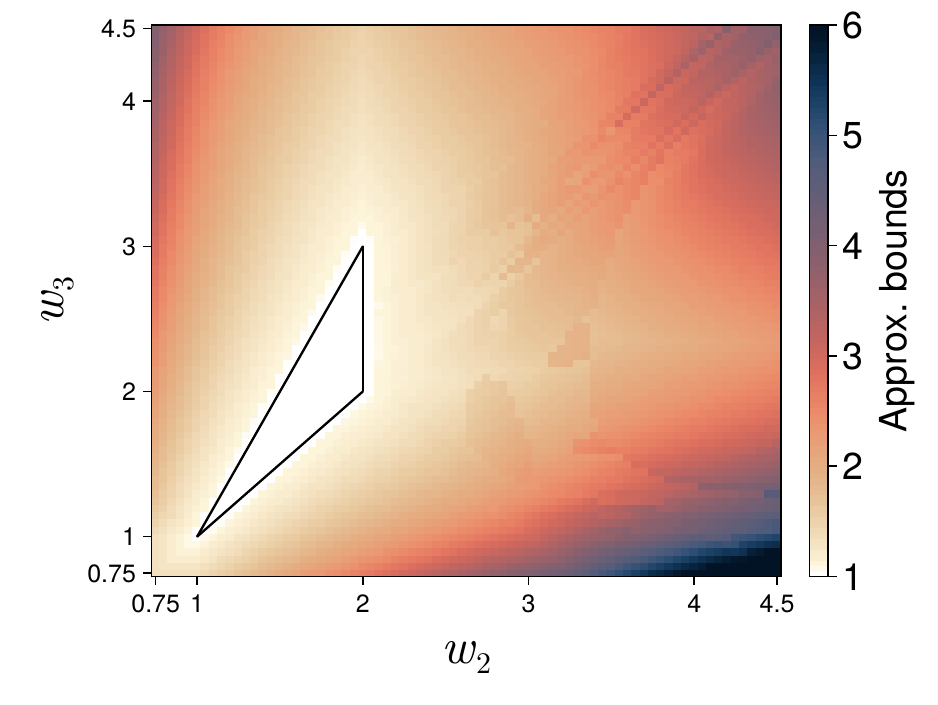} %
        \label{fig:heatmap_Linf}
    }
    \hfill
     \subfigure[Bounds obtained by our new Algorithm~\ref{alg:algo_projection}]{%
        \includegraphics[width=0.45\textwidth]{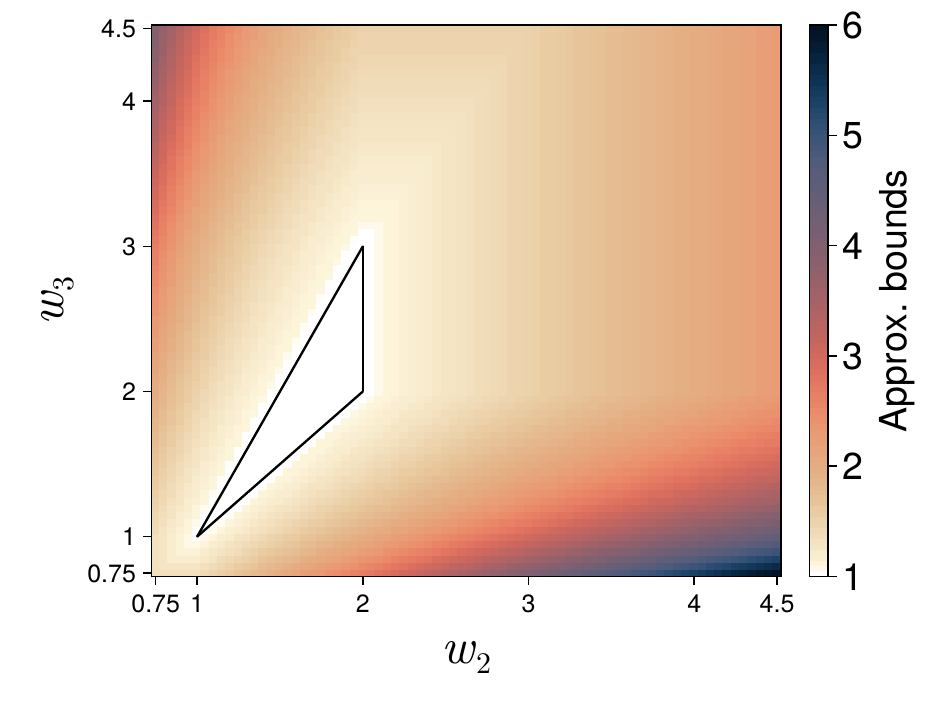} %
        \label{fig:heatmap_plcc}
    }
    \caption{Heatmaps of approximation guarantees obtained for $r$-\cbcut{} when $r \in \{6,7\}$ for a grid of $(w_2,w_3)$ choices when $w_1=1$ is fixed, using four different techniques for projecting on the submodular region. Within the submodular region (shown by a white triangle), the approximation ratio is $1$ as no projection is required. }
    \label{fig:heatmaps_norms}
\end{figure} 
Projecting an arbitrary point onto a convex region is a standard problem in computational geometry. We can use the approximation ratio in Lemma~\ref{lem:projbound} to illustrate how well these standard projection techniques allow us to approximate $r$-\cbcut{}. In more detail, we can solve a minimum-norm projection problem of the form
\begin{equation}
    \label{eq:standardproj}
    \min \quad \| \vw - \vw' \| \quad \text{ subject to } \hat{\vw}' \in \mathcal{S}_q
\end{equation}
for some choice of norm $\| \cdot \|.$ The resulting vector may not satisfy $\vw' \geq \vw,$ so we scale it to define a vector $\hat{\vw} = c \vw'$ where $c$ is as small as possible while still satisfying $\hat{\vw} \geq \vw.$ The approximation factor is then given by $\max_i \hat{w}_i/w_i.$
As an illustration, we consider $r \in \{6, 7\}$ and assess the effectiveness of finding a minimum-norm projection of an $r$-NCB vector onto the submodular space $\mathcal{S}_3$. We specifically consider the $\ell_1$-norm, $\ell_2$-norm, and $\ell_\infty$-norm, as common examples of norms one might wish to minimize. For simplicity, we set $w_1=1$ so that the problem reduces to projecting a two-dimensional point $[w_2, w_3]$ into the submodular region defined by inequalities $w_2 \leq w_3$, $1 \leq w_2 \leq 2$, and $2w_2 \geq w_3 + 1$ (blue region in Figure~\ref{fig:67cb-projection}). It is not hard to show that scaling and projecting in a 2-dimensional space in this way does not affect the resulting approximation guarantee.

The best approximation factors for a grid of $(w_2, w_3)$ points is shown using heatmaps in Figures~\ref{fig:heatmap_L1},~\ref{fig:heatmap_L2}, and~\ref{fig:heatmap_Linf}. The darkness of a point indicates how large the worst-case approximation ratio $\rho_r = \max \left(\frac{w_1}{w_1},\frac{w_2}{w_2}\right)$ is for a given point $(w_2, w_3)$. Not surprisingly, the approximation factor gets worse as we move farther from the submodular region, for all three norms. The approximation ratios achieved using $\ell_1$ and $\ell_2$ are very similar, though at close inspection, the $\ell_1$ result is always at least as good as using $\ell_2$, and can be strictly better. The $\ell_\infty$ approximation is never better than $\ell_1$ and $\ell_2$, but can be noticeably worse. Figures~\ref{fig:heatmap_L1_L2} through \ref{fig:heatmap_Linf_L2} display the difference in approximation ratios achieved when using the three different norms. This shows that $\ell_1$ projections produce the best approximation bounds among them, followed by $\ell_2$, and then $\ell_{\infty}$. However, none of these minimum-norm projection approaches produces the best approximation factor. In Figure~\ref{fig:heatmap_plcc} we show the approximation factors achieved by a new projection technique we design, which we will prove provides the optimal approximation factor that can be achieved for $r$-\cbcut{} by replacing a set of non-submodular splitting penalties with a nearby set of submodular splitting penalties. Figures~\ref{fig:heatmap_L1_Lin}, \ref{fig:heatmap_L2_Lin}, and \ref{fig:heatmap_Linf_Lin} show the difference in approximation factor between each norm-minimizing projection and our new approach. Each norm-minimizing technique produces the same result as our method in \emph{some} regions, but there are also always regions in which they fail to find the best approximation factor. 

\begin{figure}
    \centering
    \subfigure[Difference between $\ell_2$ and $\ell_1$ projections ]{
        \includegraphics[width=0.45\textwidth]{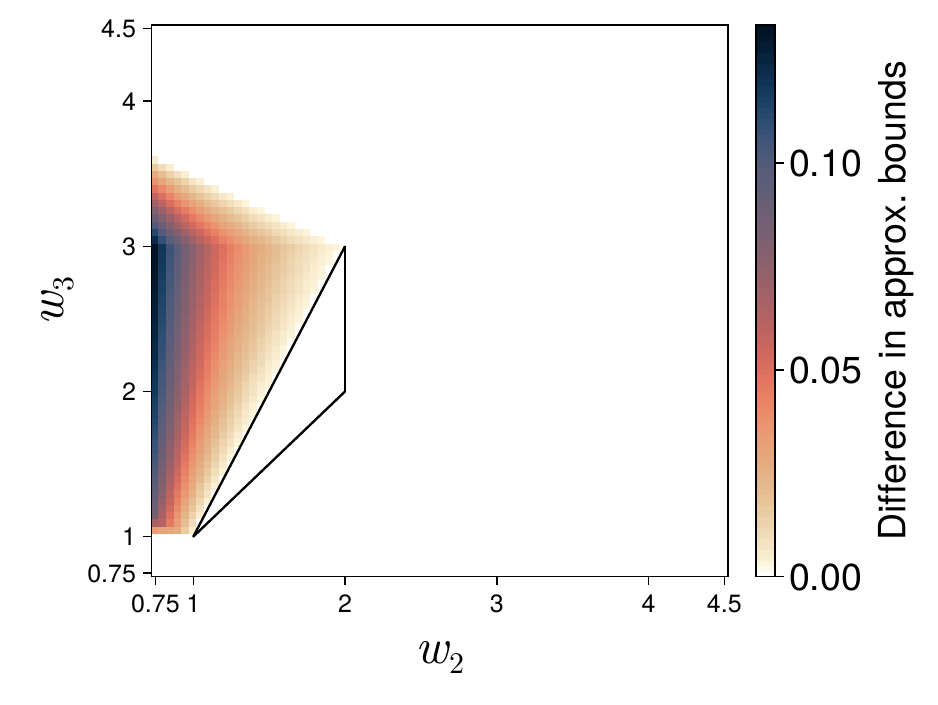}
        \label{fig:heatmap_L1_L2}
    }
    \hfill
    \subfigure[Difference between $\ell_{\infty}$ and $\ell_1$ projections]{
        \includegraphics[width=0.45\textwidth]{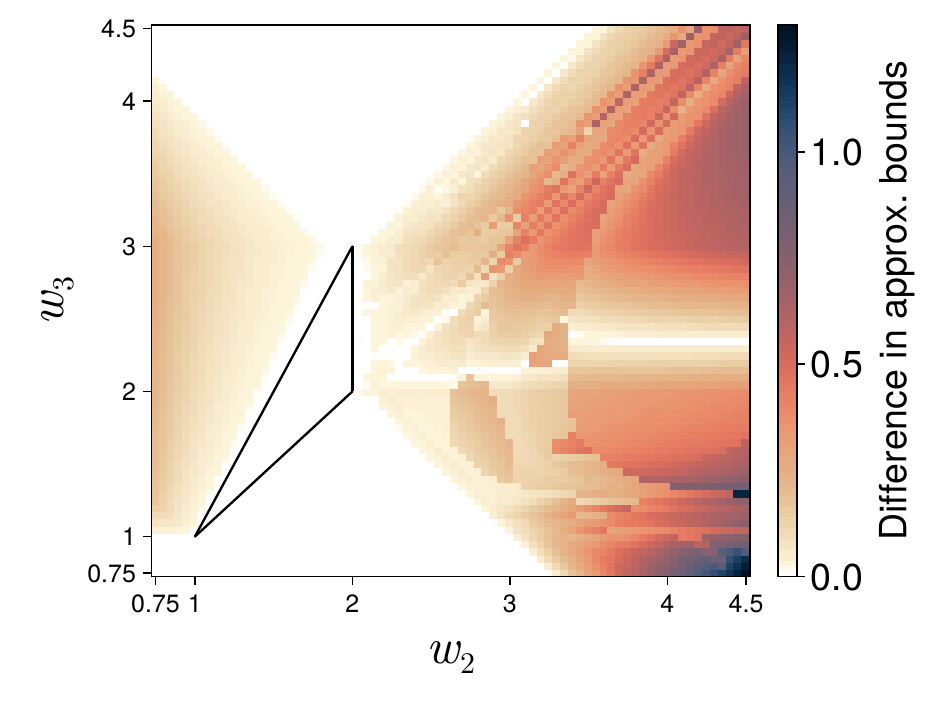} %
        \label{fig:heatmap_L1_Linf}
    }
    \hfill
    \subfigure[Difference between $\ell_{\infty}$ and $\ell_2$ projections]{%
        \includegraphics[width=0.45\textwidth]{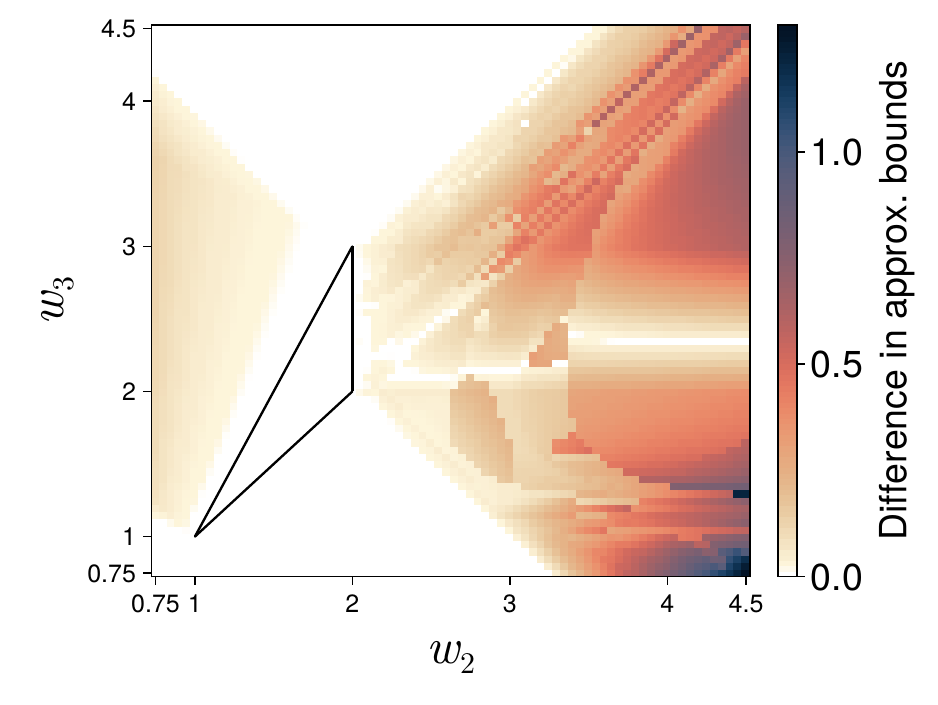} %
        \label{fig:heatmap_Linf_L2}
    }
    \hfill
    \subfigure[Difference between $\ell_1$ projection and Algorithm~\ref{alg:algo_projection}]{%
        \includegraphics[width=0.45\textwidth]{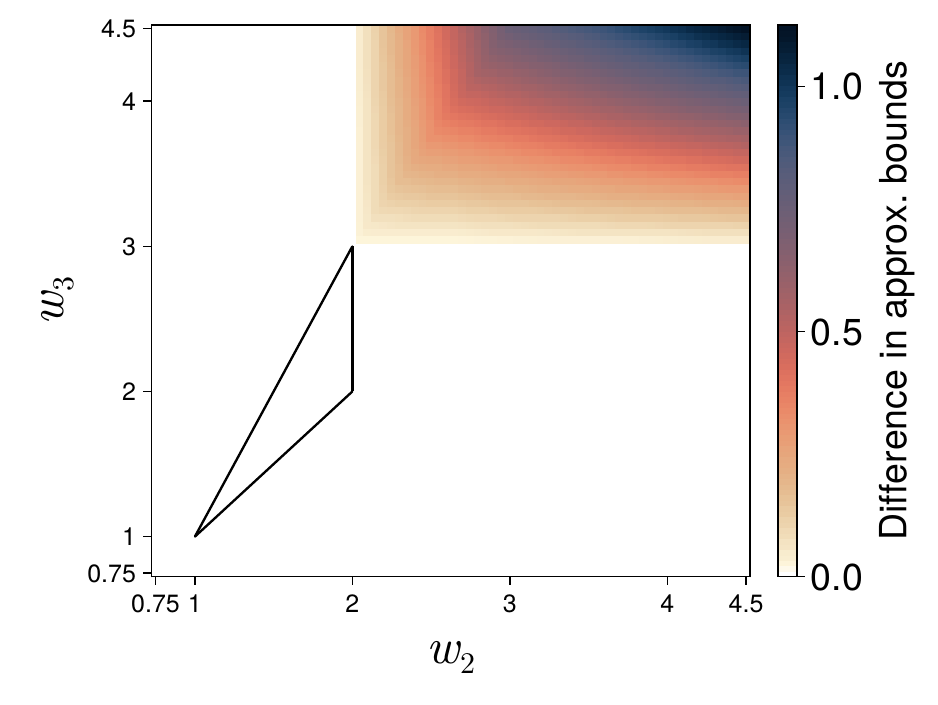} %
        \label{fig:heatmap_L1_Lin}
    }
    \hfill
     \subfigure[Difference between $\ell_2$ projection and Algorithm~\ref{alg:algo_projection}]{%
        \includegraphics[width=0.45\textwidth]{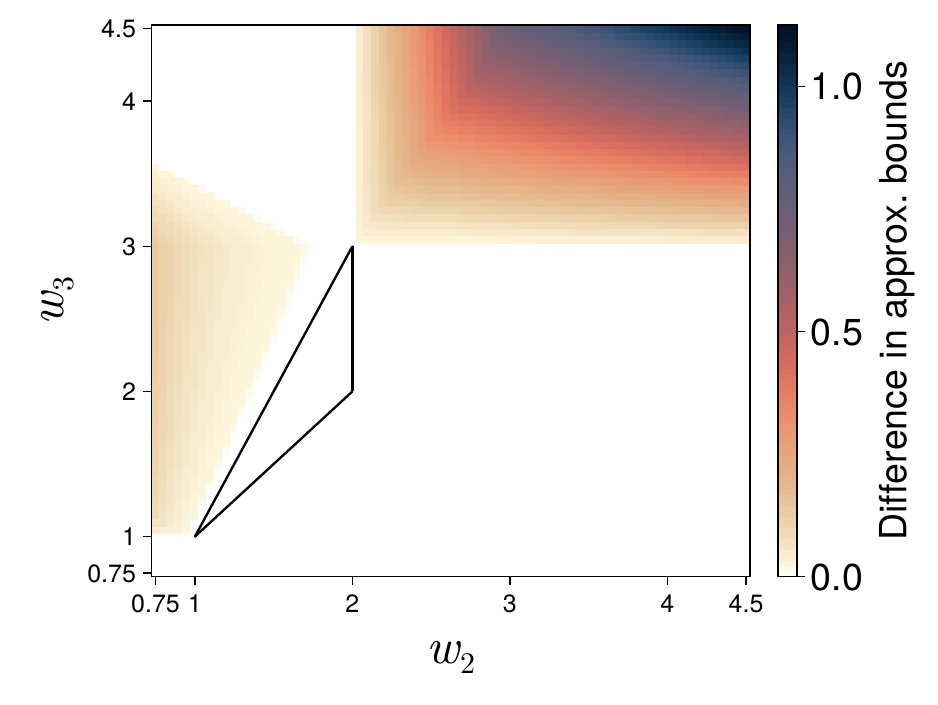} %
        \label{fig:heatmap_L2_Lin}
    }
    \hfill
     \subfigure[Difference between $\ell_{\infty}$ projection and Algorithm~\ref{alg:algo_projection}]{%
        \includegraphics[width=0.45\textwidth]{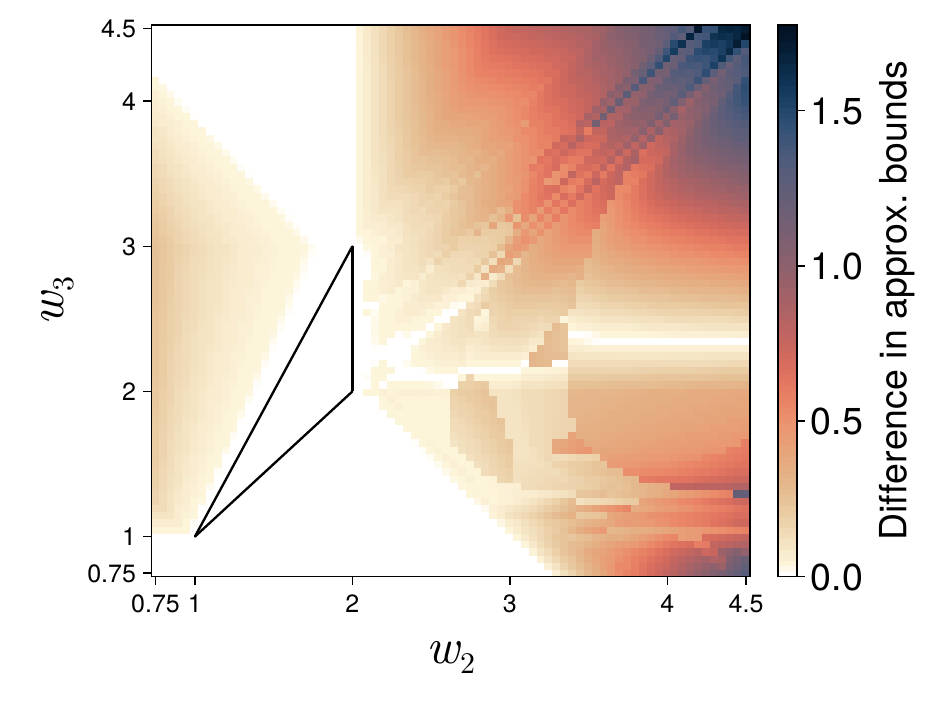} %
        \label{fig:heatmap_Linf_Lin}
    }
    \caption{Heatmaps showing the difference in approximation factors obtained using different types of projection techniques.}
    \label{fig:heatmaps_diffs}
\end{figure}

\subsection{Approximation using a piecewise-linear concave upper bound}
We now present full details for our optimal projection method, whose approximation factors are illustrated for $r \in \{6,7\}$ in Figure~\ref{fig:heatmap_plcc}. For convenience, and using a slight abuse of notation, we think of the $r$-NCB splitting vector as a discrete integer function $\vw \colon \{0,1,2,\hdots, q\} \rightarrow \mathbb{R}_+$, with range equal to the set of splitting penalties $\vw(i)=w_i$ for $i\in\{0,1,2,\hdots,q\}$. Note that when treating $\vw$ as an integer function we will explicitly encode the fact that $\vw(0) = 0$. We aim to find new splitting penalties $\hat{w}_i \ge w_i$ that satisfy the submodular constraints while minimizing the largest ratio $\frac{\hat{w}(i)}{w(i)}$ over all $i > 0$. This is the same as finding another integer function $\hat{\vw}$ that solves the following optimization problem:
\begin{align}
       \minimize_{\hat{\vw}} \quad &  \kappa  \quad \label{eq:minmax1}\\
     \text{such that}\quad & \kappa \ge \frac{\hat{\vw}(i)}{\vw(i)}  \quad \forall i \in \{1,2, \hdots, q\}  \tag{\ref{eq:minmax1}a} \label{eq:constraint_1a}\\
     &2\vhw(i) \ge \vhw(i-1) + \vhw(i+1)  \quad \forall i = 2,3,\hdots,q-1 \tag{\ref{eq:minmax1}b} \label{eq:constraint_1b}\\
     &\vhw(i+1) \ge \vhw(i) \quad \forall i=1,2,\hdots,q-1 \tag{\ref{eq:minmax1}c} \label{eq:constraint_1c}\\  
     & \hat{\vw}(i) \ge \vw(i) \quad \forall i \in \{0,1,2, \hdots, q\} \tag{\ref{eq:minmax1}d} \label{eq:constraint_1d}
\end{align}
Since $\vw$ is given, this is just a small linear program (LP). However, we do not need a general LP solver to find the solution. We present a simple approach for finding an optimal $\hat{\vw}$ by casting it as the equivalent task of finding a piecewise-linear concave (PLC) function to ``cover'' the points $(i,\vw(i))$ as tightly as possible. For this latter formulation, we begin by defining a class of piecewise-linear concave functions of interest.

\begin{definition} \label{def:PLC_function}
    Let $\mathcal{F}_q$ denote the class of PLC functions $\hat{\vf}:[0,q] \rightarrow\mathbb{R}_+$, where each function $\hat{\vf} \in \mathcal{F}_q$ satisfies the following properties:
    \begin{enumerate}[label=(\alph*)]
        \item
        $\hat{\vf}(0) = 0$.
        \item
        $\hat{\vf}$ is non-decreasing: $\hat{\vf}(x) \ge \hat{\vf}(y)$ for all $x \ge y$.
        \item For $\ell \in \{1,2, \hdots, q\}$, $\hat{\vf}$ is a linear function $\vhf_{\ell}$ on interval $[\ell-1,\ell]$, with positive slope $m_{\ell} \geq 0$ and an intercept $c_{\ell}$:
        \begin{align*}
            \hat{\vf}(x) = \hat{\vf}_{\ell}(x) = m_{\ell}\cdot x + c_{\ell} \quad \quad \text{ for $x \in [\ell -1, \ell].$}
        \end{align*}
        \item[(d)] \phantomsection\label{prop:d} $\hat{\vf}$ is concave, meaning the slopes are non-increasing: $m_{\ell} \ge m_{\ell+1}$ for all $\ell \in \{1,2,\hdots,q\}$.
    \end{enumerate}
\end{definition}


We now recast optimization Problem~\eqref{eq:minmax1} 
as the following PLC cover problem:
\begin{align}
       \minimize_{\hat{\vf}} \quad &  \kappa  \quad \label{eq:minmax2}\\
     \text{such that}\quad & \kappa \ge \frac{\hat{\vf}(i)}{\vw(i)}  \quad \forall i \in \{1,2, \hdots, q\} \tag{\ref{eq:minmax2}a} \label{eq:constraint_2a}\\
     & \hat{\vf}(i) \ge \vw(i) \quad \forall i \in \{0,1,2, \hdots, q\} \tag{\ref{eq:minmax2}b} \label{eq:constraint_2b}\\
     &\hat{\vf} \in \mathcal{F}_q. \tag{\ref{eq:minmax2}c} \label{eq:constraint_2c}
\end{align}
\begin{figure}
    \centering
    \includegraphics[width=0.4\linewidth]{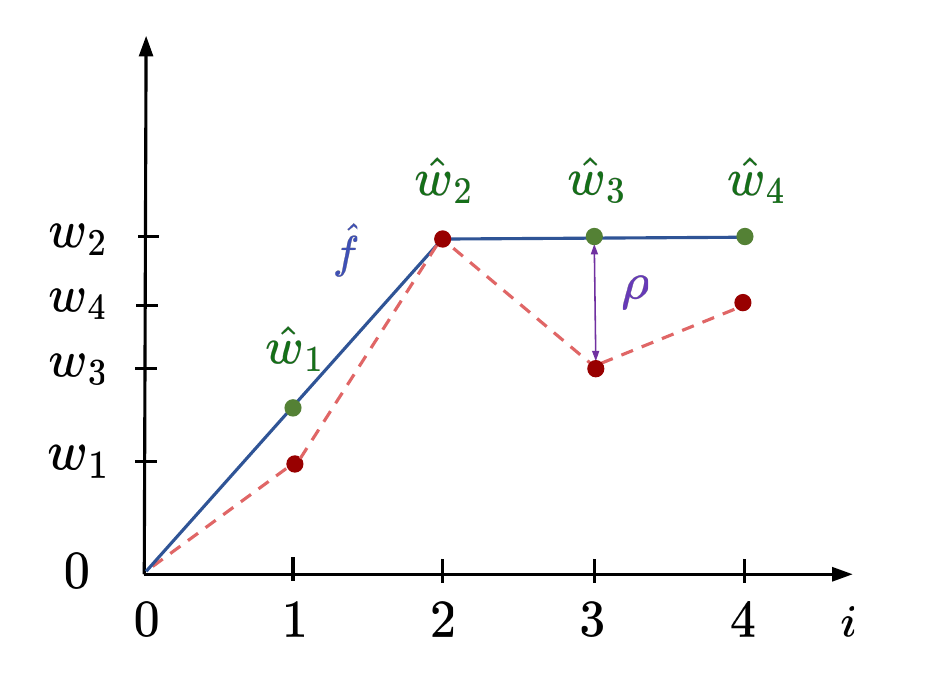}
    \caption{An optimal PLC cover $\vhf$ (shown in solid blue) for the integer function $\vw$, with the discrete values of $\vw(i)$ marked in red. The blue line segments for each interval $[i,i+1]$ represent the linear pieces of $\vhf$, while the largest gap between $\vhf(i)$ and $\vw(i)$ (here at $i=3$) is denoted by the approximation bound $\rho$. The red dashed line interpolates $\vw$; the fact that it is not a PLC function indicates that $\vw$ does not define splitting penalties for a submodular function.}
    \label{fig:plcc_ex}
\end{figure}
We refer to a feasible solution for this problem as a PLC cover for $\vw$. Figure~\ref{fig:plcc_ex} illustrates the optimal PLC cover $\hat{\vf}$ for a function $\vw$ corresponding to a set of non-submodular splitting penalties. Given a solution $\hat{\vf}$, we can define a set of submodular splitting parameters by $\hat{w}_i = \hat{\vf}(i)$ for $i \in \{0,1, \hdots, q\}$. The following lemma shows that this set of splitting parameters is the optimal projection for the splitting penalties defining $\vw$.
\begin{lemma}
    Let $\vhf \in \mathcal{F}_q$ solve Problem~\eqref{eq:minmax2}, and define $\vhw$ where $\vhw(i) = \vhf(i)$ for all $i\in \{0,1,2, \hdots, q\}$. Then, $\vhw$ solves Problem~\eqref{eq:minmax2}.
\end{lemma}

\begin{proof}
    The concavity property in Definition~\ref{def:PLC_function}(d) of $\vhf\in \mathcal{F}_q$ ensures that the slope of its linear pieces $\ell \in \{1,2, \hdots, q\}$ is non-increasing. For a given segment $\hat{\vf}_\ell$ on the interval $[\ell-1,\ell]$, the slope is give by $m_{\ell} = \vhf({\ell}) - \vhf(\ell-1)$. Similarly, for the next segment, we have $m_{\ell+1} = \vhf(\ell+1) - \vhf(\ell)$. Thus:
    \begin{align*}
        \vhf(\ell) - \vhf(\ell-1) \ge \vhf(\ell+1) - \vhf(\ell) &\implies 2\vhf(\ell) \ge \vhf(\ell-1)+\vhf(\ell+1) \implies 2\vhw(\ell) \ge \vhw(\ell-1)+\vhw(\ell+1).
    \end{align*}
This matches constraint~\eqref{eq:constraint_1b} of Problem~\eqref{eq:minmax1}. The non-decreasing function property in Definition~\ref{def:PLC_function}(b) directly satisfies constraint~\eqref{eq:constraint_1c}. Finally, since $\vhf$ is an upper bound for $\vhw$, we know
$\vhf(i) \ge \vw(i) \implies \vhw(i) \ge \vhw(i)$, satisfying 
 constraint~\eqref{eq:constraint_1d} of Problem~\eqref{eq:minmax1}. Thus, solving Problem~\eqref{eq:minmax2} produces a feasible solution to Problem \eqref{eq:minmax1}. In fact, this will produce an optimal solution to Problem~\eqref{eq:minmax1}. To prove this, assume instead that there exists some other feasible solution $\tilde{\vw}$ for Problem \eqref{eq:minmax1} that has a strictly lower objective score than $\hat{\vw}$ does. If this were true, we could construct a new function $\tilde{\vf} \in \mathcal{F}_q$ by taking the linear interpolation of the points in $\tilde{\vw}$. This would produce a strictly better solution to Problem~\eqref{eq:minmax2}, contradicting the optimality of $\hat{\vf}$.
\end{proof}

\subsection {Building an optimal PLC cover}
We solve Problem~\eqref{eq:minmax2} for an integer function $\vw$ by iteratively defining a set of linear pieces.

\textit{First set of linear pieces:} The first segment, $\vhf_1$, must pass through the origin so that $\vhf(0)=0$. The slopes of $\vhf_{\ell}$ must also be non-increasing to maintain concavity. To find $\vhf_1$ with a slope of $m_1 \ge m_{\ell}$ for every $\ell \in \{1,2, \hdots, q\}$, we consider different lines defined by connecting the origin to each point $(i,\vw(i))$ for $i \in \{1,2, \hdots, q\}$.
Among these options, we select the line going through the origin and $(i^*,\vw(i^*))$ where $i^*$ is chosen to maximize the slope $\vw(i^*)/i^*$. Choosing a line with a less steep slope would lead to a violation in the requirement that $\hat{\vf}(i) \geq \vw(i)$ for every $i \in \{1,2, \hdots, q\}$. Choosing a line with a steeper slope would be suboptimal as it leads to higher ratios $\hat{\vf}(i)/\vw(i)$ for $i \in \{1,2,\hdots,i^*\}.$ Thus, selecting the line through the origin with slope $\vw(i^*)/i^*$ gives the tightest upper bound for $\vw$ over the interval $[0,i^*]$, and gives us all line segments for $\hat{\vf}$ in that interval:
\begin{equation*}
    \vhf_{\ell}(x) = \frac{\vw(i^*)}{i^*}\cdot x \quad \text{ for every } \ell \in \{1,2, \hdots, i^*\}.
\end{equation*}

\textit{Iteratively constructing the next linear pieces:}  Consider a general setting where we are given a starting integer $t$ and we assume we have a set of linear pieces for an optimal $\vhf$ over the interval $[0,t]$, with $\vhf(t) = \vw(t)$. We then wish to find the next set of segments $\vhf_{\ell}$ that bounds $\vw$ in the interval $[t,q]$. We first check whether $\vhf(t) = \vw(t) \geq \vw(i)$ for every $i \in [t,q]$. If so, then we choose the minimum possible slope of 0 (to ensure $\vhf$ is an increasing function) and define $\hat{\vf}(x) = \vhf(t)$ for $x \in [t,q]$. 
If not, similar to our approach for the first linear pieces, we consider all lines between $(t,\vw(t))$ and every other point $(i,\vw(i))$ for $i\in [t+1,q]$ for which $\vw(i) > \vw(t)$. By construction, all of these linear pieces are less steep that the lines used in defining $\vhf$ over the interval $[0,t]$. We then select a new line with the steepest slope $m^*$ among our choices for the next linear piece. Assume this connects point $(t,\vw(t))$ to $(i^*,\vw(i^*))$ for some $i^* \in [t+1,q]$. The linear pieces defining $\vhf_{\ell}$ for $\ell \in [t+1,i^*]$ are then given by:
\begin{align*}
    \vhf_{\ell}(x) = \frac{\vw(i^*) - \vw(t)}{i^* - t}(x-t) + \vw(t) \quad \text{ for every } \ell \in [t+1,i^*].
\end{align*}
At this point, we have the tightest upper bound for $\vw$ over the interval $[0,i^*]$, and we apply the same steps recursively to the remaining interval $[i^*,q]$. This greedy procedure continues until a point where we add a linear piece with slope zero, or until $i^* = q$ for a given iteration. The pseudocode for building the nearest PLC cover using this greedy approach is outlined in Algorithm~\ref{alg:algo_projection}.  We summarize with the following lemma.
\begin{lemma}
    Algorithm~\ref{alg:algo_projection} solves Problem~\eqref{eq:minmax1} by constructing an optimal PLC cover $\vhf:[0,q] \rightarrow \mathbb{R}_+$ for Problem~\eqref{eq:minmax2}.
\end{lemma}

\begin{algorithm}[t]
	\caption{Solve Problem~\eqref{eq:minmax1} via PLC cover}
 	\label{alg:algo_projection}
	\begin{algorithmic}
		\STATE{\bfseries Input:} Integer function $\vw \colon \{0,1,2, \hdots, q\} \rightarrow \mathbb{R}_+$ with $\vw(0) = 0$ 
		\STATE {\bfseries Output:} $\hat{\vw}$ solving Problem~\eqref{eq:minmax1}
        \STATE Initialize $\hat{\vw}(i) = 0$ for every $i \in \{0,1, \hdots, q\}$
        \STATE Initialize $t=0$
        \WHILE{$ t \le q$} 
        \STATE Define $m_i  = \frac{\vw(i) - \vw(t)}{i - t}$ for $i \in [t+1,q]$
            \IF{ $\max_{i \in [t+1,q]} m_i \leq 0$}
            \STATE $i^* = q, m^* = 0$ \hfill \texttt{// flat line to end}
            \ELSE
            \STATE $i^* = \argmax_{i \in [t+1,q]} m_i$
            \STATE $m^* = m_{i^*}$ \hfill \texttt{// positive slope linear piece}
            \ENDIF
            \FOR{$i = t+1$ to $i^*$}
            \STATE $\hat{\vw}(i) = m^*(i-t) + \vw(t)$
            \ENDFOR
            \STATE $t \leftarrow i^*$
            \ENDWHILE
		\STATE {\bfseries Return:} $\hat{\vw}$
	\end{algorithmic}
\end{algorithm}


        \section{Approximation Hardness Results}
\label{sec:approxhard}
We now prove harness results to strongly indicate that our projection-based technique for non-submodular \cbcut{} produces the best approximation guarantees we can hope for.
We focus specifically on $4$-\cbcut{} when $w_1 = 1$, as this represents the most basic type of NP-hard non-submodular case.
We first show that the reductions from \textsc{MaxCut} (Section~\ref{sec:4cbcut}) not only imply NP-hardness, but also APX-hardness in the non-submodular region $w_2 \notin [1,2]$. This means that for \textsc{CBcut}$(4,\{w_1 = 1, w_2\}$ with $w_2 \notin [1,2]$, there exists some constant $c > 1$ such that it is NP-hard to approximate the problem below a factor $c$. More significantly, we prove that assuming the Unique Games Conjecture, the problem is hard to approximate below the approximation factors obtained by projecting non-submodular penalties to the submodular region.

\subsection{APX-hardness via Reduction from \textsc{MaxCut}}
Let $G = (V,E)$ represent an instance of \textsc{MaxCut} where $m = |E|$ and where we use $k^*$ to denote the optimal number of edges that are cut. For our APX-hardness results, we will use the fact that there is always a way to cut at least half the edges in a graph, so $k^* \geq m/2$. We also use the fact that \textsc{MaxCut} is NP-hard to approximate to within a factor better than $16/17$~\cite{hastad}.

\begin{lemma}
    For fixed $w_2 > 2$, i is NP-hard to approximate \cbf{} to within a factor smaller than 
    \begin{equation*}
        1 + \frac{w_2 - 2}{17w_2 + 34}.
    \end{equation*}
\end{lemma}
\begin{proof}
Consider the reduction from \textsc{MaxCut} to \cbf{} (the problem variant with pairwise edges) for $w_2 > 2$ in Figure~\ref{fig:4cbcut}. Recall that we fix $w_1 = 1$. 
A bipartition of the nodes in $V$ that cuts $k$ edges corresponds to an $s$-$t$ cut in the reduced hypergraph with cut value
$2k + (m-k)w_2$.
Consider a hypergraph $s$-$t$ cut algorithm that produces a cut with value $2k'+(m-k')w_2$, where $k'\le k^*$ is the number of gadgets where the nodes $(u,v) \in E$ defining the gadget are split into different sides of the $s$-$t$ cut. If the algorithm is a $\beta$-approximation algorithm, then
\begin{align*}
\label{eq: maxcut1} 
   2k'+(m-k')w_2 &\leq \beta (2k^*+(m-k^*)w_2) \\
   \implies mw_2 + (2-w_2)k' &\leq \beta m w_2 + \beta(2-w_2)k^* \\
   \implies mw_2 + \beta(w_2 - 2)k^* &\leq \beta m w_2 + (w_2-2)k'\\
   \implies \beta(w_2-2)k^* &\leq (\beta - 1)mw_2 + (w_2-2)k'\leq (\beta - 1)2k^*w_2 + (w_2 - 2)k',
\end{align*} 
where in the last step we have used the fact that $k^* \ge m/2$.
Combining this with the fact that $k'/k^* \leq 16/17$ leads to a lower bound on the best approximation factor we can hope to achieve:
\begin{align*}
    k^*[\beta (w_2 - 2) -2(\beta -1)w_2] &\le (w_2 - 2)k'   \\
    \implies \beta(w_2 - 2) - 2(\beta -1)w_2 &\le {\frac{k'}{k^*}}(w_2-2) \leq \frac{16}{17}(w_2 - 2)\\
    \implies \beta(w_2 -2 -2w_2) & \le \frac{16w_2 - 32 -34w_2}{17}\\
    \implies \beta(2+w_2) &\ge \frac{18w_2 + 32}{17}\\
    \implies \beta &\ge \frac{18w_2 + 32}{17(2+w_2)}= \frac{18w_2 + 32}{17w_2 + 34} = 1+ \frac{w_2 - 2}{17w_2 + 34}.
 \end{align*}
\end{proof}
We prove a similar result when $w_2 < 1$.
\begin{lemma}
    For fixed $w_2 < 1$, it is NP-hard to approximate \cbf{} to within a factor smaller than 
    \begin{equation*}
        1 + \frac{1-w_2}{17(1+w_2)}.
    \end{equation*}
\end{lemma}
\begin{proof}
Consider the reduction from \textsc{MaxCut} to \cbf{} for $w_2 < 1$ in Figure~\ref{fig:4cbcut}. Cutting $k$ edges in the \textsc{MaxCut} instance $G = (V,E)$ corresponds to a hypergraph $s$-$t$ cut value of $kw_2 + (m-k)$ (recall that $m = |E|$). Assume we have a $\beta$-approximation algorithm for the \cbf{} problem. Let $k'$ be the number of gadgets where the algorithm separates the nodes $(u,v)$ defining the gadget. We know then that
\begin{align*}
    k'w_2 + m - k' &\le \beta(k^*w_2 + m - k^*)\\
    \implies k'(w_2 - 1) + m &\le \beta k^*(w_2 - 1) + \beta m\\
    \implies \beta k^*(1-w_2) &\le k'(1-w_2) + (\beta - 1)m\\
      &\le k'(1-w_2) + 2(\beta - 1)k^* & \text{(using $k^* \ge m/2$)} \\ 
    \implies k^*[\beta(1- w_2) - 2(\beta -1)] &\le k'(1-w_2)\\
    \implies {(\beta(1- w_2) - 2(\beta -1))} &\le \frac{16}{17}(1-w_2) & \text{(using ${k'}/{k^*} \leq 16/17$)}\\
    \implies \beta(1-w_2-2) &\le \frac{16}{17}(1-w_2) -2\\
    \implies \beta(-w_2 - 1) &\le \frac{-18 - 16w_2}{17}
    \end{align*}
Thus
\begin{equation*}
    \beta(1+w_2) \ge \frac{18+16w_2}{17}\implies \beta \ge 1 + \frac{1-w_2}{17(1+w_2)}.
\end{equation*}
\end{proof}
These lemmas tell us that for any fixed $w_2 \notin [1,2]$, there exists an $\epsilon > 0$ (depending on $w_2$) such that approximating $4$-\cbcut{} to within a factor $1+\epsilon$ is NP-hard. As $w_2$ gets further from the submodular region (i.e., taking a limit $w_2 \rightarrow 0$ or $w_2 \rightarrow \infty$), the approximation factor gets worse. However, even for the most extreme values of $w_2$, this only rules out the possibility of obtaining an approximation better than ${18}/{17}$. In contrast, our best approximation factors become arbitrarily bad as $w_2$ goes to zero or infinity. We would like, therefore, to tighten this gap to show that the best approximation factors also get arbitrarily bad as we get further from submodularity.

\subsection{Inapproximability via Unique Games Conjecture}
To prove stronger approximation results, we again leverage the connection to Valued Constraint Satisfaction Problems. Every VCSP permits a linear programming (LP) relaxation known as the Basic LP~\cite{thapper2012power}, which we will define shortly for VCSPs corresponding to hypergraph $s$-$t$ cut problems. This LP includes an LP variable $x_{v,\ell} \in [0,1]$ for every VCSP variable $v$ and possible assignment $\ell$ for $v$. For Boolean VCSPs, $\ell \in \{0,1\}$, though the Basic LP is also defined for non-Boolean VCSPs. The other variables and constraints in the LP are formulated in such a way that the optimal solution to the LP lower bounds the optimal solution for the VSCP instance. The optimal value for $x_{v,\ell}$ can be interpreted as the probability of assigning variable $v$ to label $\ell$. 

Ene et al.~\cite{ene2013local} proved that when the constraint language includes the \emph{not-all-equal} predicate, the integrality gap of this LP relaxation lower bounds the best possible approximation factor for the corresponding valued constraint language, assuming the Unique Games Conjecture. In more detail, the \emph{not-all-equal} predicate on two variables is defined as
\begin{align*}
\text{NAE}_2(x,y) =
\begin{cases} 
      0 & \text{if } x = y \\
      1 & \text{if } x \neq y.
\end{cases}
\end{align*}
When viewing variables as nodes in a graph, this corresponds to a cut function for a single edge: the penalty is one if two nodes are separated and is 0 otherwise (this can be scaled by a non-negative value for weighted cut functions).
We will leverage this result to prove new UGC-hardness results for approximating $4$-\cbcut{} \textsc{with edges} when $w_1 = 1$ and $w_2 \notin [1,2]$.

\paragraph{The Basic LP for hypergraph $s$-$t$ cut problems.}
Having established the relationship between VCSPs and generalized hypergraph $s$-$t$ cut problems in Sections~\ref{sec:prelims} and~\ref{sec:nphard}, we will go back and forth between the two views interchangeably. For a generalized hypergraph $s$-$t$ cut problem on $\mathcal{H} = (\V, \E)$ (whether or not splitting functions are cardinality-based), the Basic LP relaxation is given by
\begin{align}
\label{eq:basiclp}
\text{minimize} \quad & \sum_{e \in \E} \delta_e \cdot \sum_{A \subseteq e} y_{e,A} \cdot \vw_e(A) \quad & \\
\label{lp1}
\text{subject to} \quad &\forall e \in \E , \forall v \in e\colon \quad x_{v,s} = \sum_{A \subseteq e : v \in A} y_{e,A}  \\
\label{lp2}
&\forall e \in \E , \forall v \in e\colon \quad  x_{v,t} = \sum_{A \subseteq e : v \in e \setminus A} y_{e,A}  \\
&\forall e\in \E,\forall A \in e\colon \quad 0\le y_{e,A}\le 1\\
& \forall v \in \V \colon \quad x_{v,s} + x_{v,t} = 1  \\
& \forall v \in \V \colon \quad 0 \le x_{v,s} \le 1 \text{ and } 0 \le x_{v,t} \le 1 \\
& x_{s,s} = 1 \text{ and } x_{t,t} = 1
\end{align} 
Here, $\delta_e \geq 0$ is a scalar weight associated with hyperedge $e \in \E$. This Basic LP for hypergraph $s$-$t$ cut problems can be easily derived from the more general presentation of the Basic LP in Section 4 of Ene et al.~\cite{ene2015local} (the online full version of an earlier conference paper~\cite{ene2013local}). We use variables $x_{v,s}$ and $x_{v,t}$ for each node $v \in \V$ to indicate the fractional assignment of $v$ to the $s$ and $t$ sides. If we restricted variables to be binary, then the solution would exactly be the optimal solution for the generalized hypergraph $s$-$t$ cut problem. Hence, the solution to the LP lower bounds the optimal $s$-$t$ cut value. The hardness result of Ene et al.~\cite{ene2015local} for VCSPs translates to the following result for the above Basic LP for generalized hypergraph $s$-$t$ cut problems. 
\begin{theorem}
\label{thm:ene}
    Consider the generalized hypergraph $s$-$t$ cut problem for a class of splitting functions that includes the standard cut function for size-2 hyperedges. Assume there exists some instance $\mathcal{H}$ such that the optimal $s$-$t$ cut for this instance has value $\mathit{OPT}(\mathcal{H})$ and the optimal solution to the Basic LP in~\eqref{eq:basiclp} is $\mathit{LP}(\mathcal{H})$. It is UGC-hard to approximate this class of hypergraph $s$-$t$ cut problems to within a factor strictly smaller than $\mathit{OPT}(\mathcal{H})/\mathit{LP}(\mathcal{H})$.
\end{theorem}
This result is not new, it is a special case of Theorem 4.3 from the work of Ene et al.~\cite{ene2015local}. We have adapted terminology and notation to apply to the special case of VCSPs corresponding to generalized hypergraph $s$-$t$ cut problems. 

The $4$-\cbcut{} \textsc{with edges} problem is equivalent to the constraint language with cost functions $\{\phi_4, \phi_s, \phi_t, \phi_{st}, \text{NAE}_2\}$. Equivalently, this is a class of hypergraph $s$-$t$ cut problems with two splitting functions: the standard graph cut function (corresponding to $\text{NAE}_2$), and the 4-node cardinality-based splitting function with penalties $w_1 = 1$ and $w_2$. Since this includes the standard graph cut function, Theorem~\ref{thm:ene} applies. In order to prove UGC-hardness results, we just need to show a relevant integrality gap example for the Basic LP.

\paragraph{Integrality gap for $w_2 < 1$.}
\begin{figure}
    \centering
    \subfigure[Optimum solution with $OPT^* = 1$]{%
    \centering
        \includegraphics[width=0.4\linewidth]{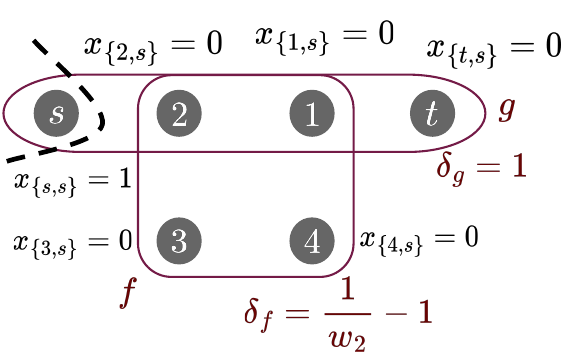}
        \label{fig:igap_opt_case1}
    }
    \hfill
    \subfigure[Basic LP solution with value $w_2$]{%
    \centering
        \includegraphics[width=0.42\linewidth]{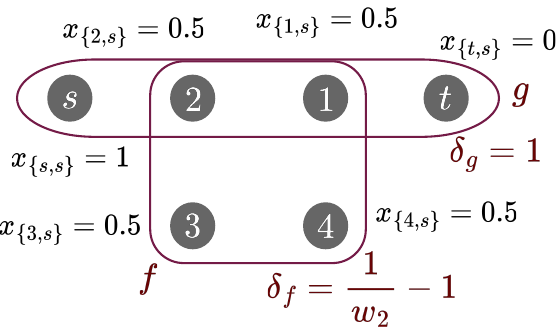} 
        \label{fig:igap_blp_case1}
    }
    \caption{Integrality gap instance of $4$-CB-cut when $w_2< w_1 = 1$.}
    \label{fig:igap_case1}
\end{figure}

Consider an instance of $4$-\cbcut{} with $w_1 = 1$ and $w_2 < 1$ defined by a hypergraph $\mathcal{H} = (\V ,\E)$ with six nodes $\V = \{1,2,3,4,s,t\}$ and two hyperedges, shown in Figure \ref{fig:igap_case1}. The first hyperedge $g=\{s,t,1,2\} \in \E$ has a weight $\delta_g = 1$, while the second hyperedge $f=\{1,2,3,4\} \in \E$ has weight $\delta_f = \frac{1}{w_2} - 1$. A minimum $s$-$t$ cut has a value of $\mathit{OPT}(\mathcal{H}) = 1$, which can be obtained by placing node $s$ on its own, $t$ on its own, or by cutting both $e$ and $f$ in an even $(2,2)$-split. Figure \ref{fig:igap_opt_case1} illustrates an optimal solution where node $s$ is placed by itself, along with the binary feasible variables for the Basic LP that represents this $s$-$t$ cut.

The Basic LP has a feasible fractional solution for the given instance where for each non-terminal node $v$ we have $x_{v,s} = x_{v,t} = 0.5$ (see Figure~\ref{fig:igap_blp_case1}). For variables $y_{f,A}$ where $A \subseteq f$, we can define $y_{f,\emptyset} = y_{f,\{1,2,3,4\}} = 0.5 $ and set $y_{f,A} = 0$ for every $A \notin \{\emptyset, f \}$.  For node $v \in \{1,2,3,4\}$ we have
\begin{align*}
    x_{v,s} &= 0.5 = y_{f,\{1,2,3,4\}} = \sum_{A \subseteq f \colon v \in A} y_{A,f} \\
    x_{v,t} &= 0.5 = y_{f, \emptyset} = \sum_{A \subseteq f \colon v \in f\backslash A} y_{A,f},
\end{align*}
so we see that constraints in~\eqref{lp1} and~\eqref{lp2} are satisfied for hyperedge $f$. For $ g = \{s,t,1,2\}$, set $y_{g,\{s,1\}} = y_{g,\{s,2\}} = 0.5$ and $y_{f,A} = 0$ for every $A \notin \{\{s,1\}, \{s,2\} \}$. We can confirm that the constraints in~\eqref{lp1} and~\eqref{lp2} are also satisfied for $g$:
\begin{align}
\begin{array}{ll}
        x_{1,s} = y_{g,\{s,1\}} + 0 = 0.5  &\quad x_{1,t} = y_{g,\{s,2\}} + 0 = 0.5\\
        x_{2,s} = y_{g,\{s,2\}} + 0 = 0.5  &\quad x_{2,t} = y_{g,\{s,1\}} + 0 = 0.5\\
        x_{s,s} = y_{g,\{s,1\}} + y_{g,\{s,2\}} = 1.0  &\quad x_{s,t} = 0 \\
        x_{t,t} = y_{g,\{s,1\}} + y_{g,\{s,2\}} = 1.0   &\quad x_{t,s} = 0.
        \end{array}
    \end{align}
Observe that $y_{f,A}\cdot\vw_f(A) = 0$ for every $A\notin \{\emptyset,f\}$ and $y_{g,A}\cdot\vw_g(A) = 0$ for every $A \notin \{\{s,1\},\{s,2\}\}$. The LP value for this feasible solution (which is in fact optimal for the LP) is therefore given by:
\begin{align*}
    \mathit{LP}(\mathcal{H}) &=\delta_f(y_{f,\emptyset} \cdot \vw_f(\emptyset)  + y_{f,\{1,2,3,4\}} \cdot \vw_f(\{1,2,3,4\})) + \delta_g(y_{g,\{s,1\}}\cdot w_2 + y_{g,\{s,2\}} \cdot w_2)\\
    &= \delta_g(w_2) = w_2 
\end{align*}
The gap between the integral and fractional solution is therefore
\begin{equation*}
    \frac{\mathit{OPT}(\mathcal{H})}{\mathit{LP}(\mathcal{H})} = \frac{1}{w_2}.
\end{equation*}
Thus, the Basic LP integrality gap is at least $\frac{1}{w_2}$, and we have the following corollary of Theorem~\ref{thm:ene}.
\begin{corollary}
    Assuming the Unique Games Conjecture, the $\textsc{CB}(4, \{w_1 = 1, w_2 <1\})$ problem cannot be approximated to within a factor better than $1/w_2$.
\end{corollary}
This lower bound matches the approximation we get by projecting to the nearest submodular penalties ($\hat{w}_1 = 1, \hat{w}_2 = 1)$, showing that this simple projection is optimal assuming UGC.

\paragraph{Integrality gap for $4$-CB-cut for $w_2 > 2$.}
\begin{figure}
    \centering
    \subfigure[Optimal solution with $s$-$t$ cut value $w_2$]{%
        \includegraphics[width=0.45\linewidth]{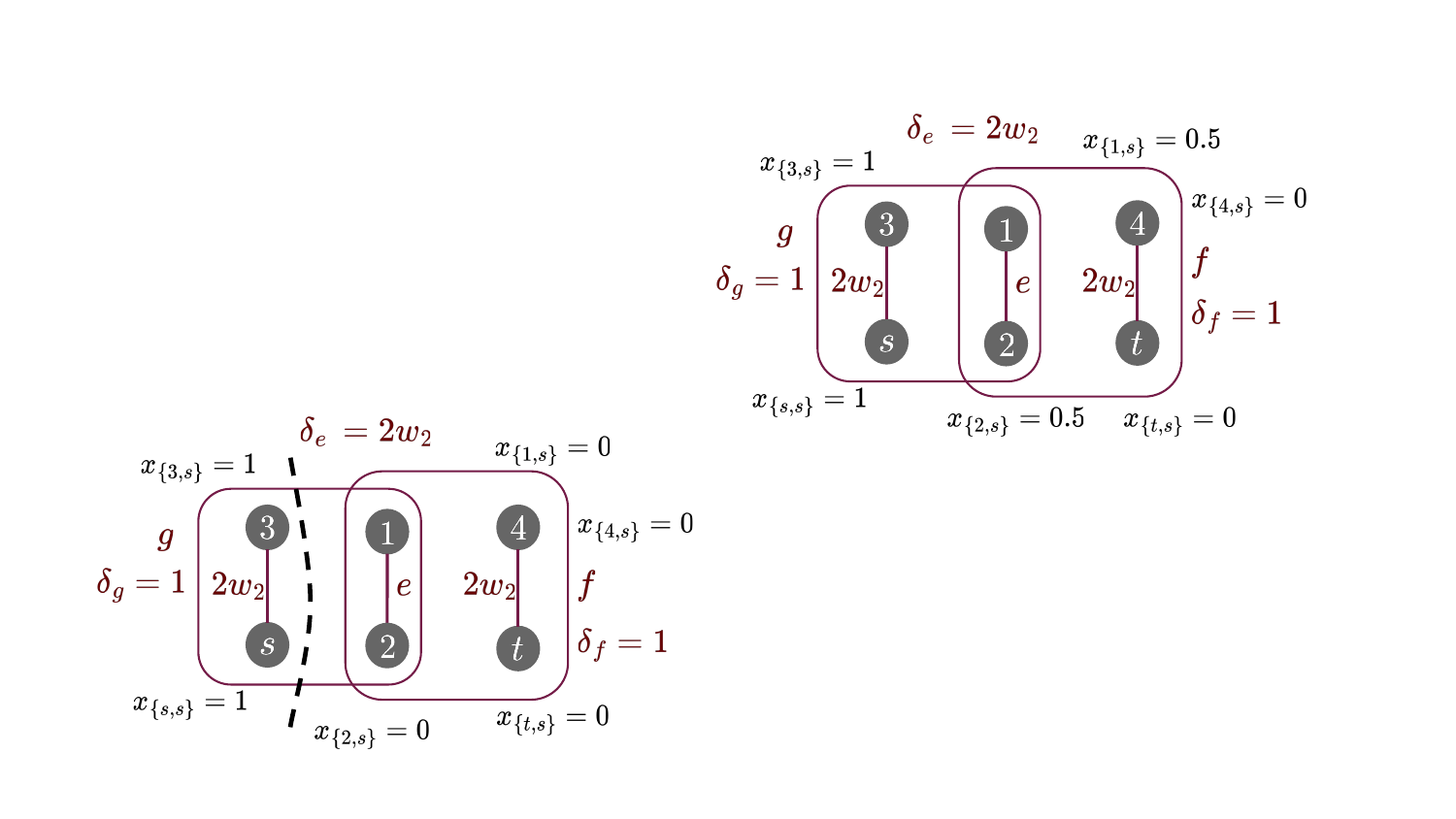}
        \label{fig:igap_opt_case2}
    }
    \hfill
    \subfigure[Basic LP solution with objective value $2$]{%
        \includegraphics[width=0.45\linewidth]{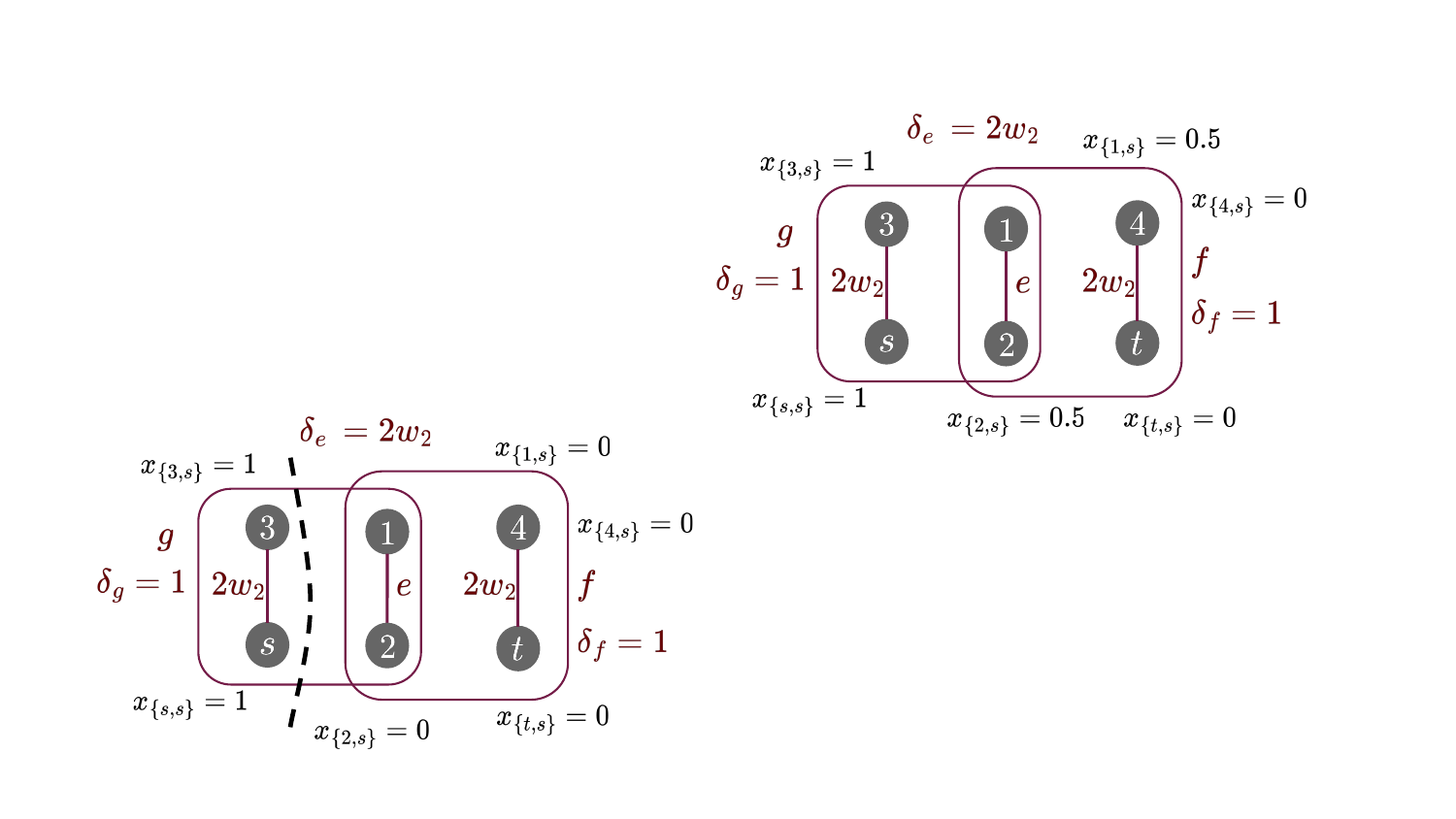} 
        \label{fig:igap_blp_case2}
    }
    \caption{Instance for 4-CB-cut when $w_2>2w_1$}
    \label{fig:igap_case2}
\end{figure}

Consider an instance of $4$-\cbcut{} given by the hypergraph $\mathcal{H} = (\V,\E)$ with six nodes $\V = \{1,2,3,4,s,t\}$ and five hyperedges, shown in Figure \ref{fig:igap_case2}. The hyperedges $f,g \in \E$ have weight $\delta_g = \delta_f = 1$, and edges $\{s,3\}$, $\{t,4\}$, and $\{1,2\}$ have weight $2w_2$. 
%
The minimum $s$-$t$ cut solution has a cut value of $\mathit{OPT}(\mathcal{H}) = w_2$, which can be achieved by cutting either hyperedge $g$ or $f$ in an even $(2,2)$ split. One such division with cut set $S = \{s,3\}$ is illustrated in Figure \ref{fig:igap_opt_case2}, along with binary feasible variables for this solution. The Basic LP has a feasible solution where $x_{3,s} = x_{4,t} = 1$, $x_{4,s} = x_{3,t} = 0$, and $x_{1,s} = x_{1,t} = x_{2,s} = x_{2,t} = 0.5$. 

For $g = \{s,1,2,3\}$, the edge variables are $y_{g,\{s,2,3\}} = y_{g,\{s,1,3\}} = 0.5$ and $y_{g,A} = 0$ for every $A\notin \{\{s,2,3\},\{s,1,3\}\}$. 
We can confirm that the constraints in~\eqref{lp1} and~\eqref{lp2} are satisfied:
\begin{align*}
\begin{array}{ll}
    x_{1,s} = y_{g,\{s,1,3\}} = 0.5 &\quad x_{1,t} = y_{g,\{s,2,3\}}= 0.5\\
    x_{2,s} = y_{g,\{s,2,3\}} = 0.5 &\quad x_{2,t} = y_{g,\{s,1,3\}} = 0.5\\
    x_{3,s} = y_{g,\{s,1,3\}} + y_{g,\{s,2,3\}} = 1 &\quad x_{3,t} = 0 \\
    x_{s,s} = y_{g,\{s,1,3\}} + y_{g,\{s,2,3\}} = 1 & \quad x_{s,t} = 0.
\end{array}
\end{align*}
For $f = \{t,1,2,4\}$, set $ y_{f,\{2\}} = y_{t,\{1\}} = 0.5$ and $y_{f,A} = 0$ for $A \notin \{\{2\},\{1\}\}$. We check constraints:
\begin{align*}
    \begin{array}{ll}
        x_{1,s} = y_{f,\{1\}} = 0.5 &\quad x_{1,t} = y_{f,\{2\}} = 0.5\\
        x_{2,s} = y_{f,\{2\}} = 0.5 &\quad x_{2,t} =  y_{f,\{1\}} = 0.5\\
        x_{4,s} = 0  &\quad x_{4,t} = y_{f,\{2\}} + y_{f,\{1\}} = 1\\
        x_{t,s} = 0 &\quad x_{t,t} =  y_{f,\{2\}} + y_{f,\{1\}}=1.
    \end{array}
\end{align*}
For edges $\{s,3\}$ and $\{t,4\}$, set $y_{\{s,3\},\{s,3\}} = 1$ and $y_{\{t,4\},\{t,4\}} = 1$. These also satisfy constraints since $x_{3,s} = 1$ and $x_{4,t} = 1$. For $e = \{1,2\}$, set $y_{e,e} = y_{e,\emptyset} = 0.5$, satisfying constraints:
 \begin{align*}
 \begin{array}{ll}
     x_{1,s} = y_{e,e} = 0.5 &\quad x_{1,t} = y_{e,\emptyset} = 0.5\\
     x_{2,s} = y_{e,e} = 0.5 &\quad x_{2,t} = 0.5.
 \end{array}
 \end{align*}
The LP value associated with this feasible solution is given by
\begin{align*}
    \mathit{LP}(\mathcal{H}) &= \delta_g \left[y_{g,\{s,2,3\}}  \vw_g(\{s,2,3\}) + y_{g,\{s,1,3\}} \vw_g(\{s,1,3\})\right] + \delta_f\left[y_{f,\{1\}} \vw_f(\{1\}) + y_{f,\{2\}}  \vw_f(\{2\}\right]\\
    &=\delta_g \cdot w_1 + \delta_f \cdot w_1 = 2.
\end{align*}
Consequently, the integrality gap is $\mathit{OPT}(\mathcal{H})/\mathit{LP}(\mathcal{H}) = w_2 /2$.
\begin{corollary}
    Assuming the Unique Games Conjecture, the $\textsc{CB}(4, \{w_1 = 1, w_2 > 2\})$ problem cannot be approximated to within a factor better than $w_2/2$.
\end{corollary}

      \section{Conclusion}
        This paper draws a connection between generalized hypergraph cut problems and Valued Constraint Satisfaction Problems, providing a simple way to settle a recently posed question on the tractability of non-submodular cardinality-based hypergraph $s$-$t$ cut problems~\cite{aksoy2023seven}. Our results confirm that the latter problem is NP-hard for all non-submodular parameter choices except for a degenerate case where a zero-cost solution is easy to achieve. Our main contribution is to provide the first systematic study of approximation algorithms for non-submodular cardinality-based $s$-$t$ cuts. In particular, we developed a greedy algorithm that provides an optimal way to project a non-submodular splitting function to the submodular region. We complemented this by proving---in the case of 4-node hyperedges---that the approximation factors achieved by this projection method are the best possible assuming the Unique Games Conjecture. One open direction is to pursue stronger hardness of approximation results that do not rely on UGC. In particular, our APX-hardness results only show that it is NP-hard to achieve an approximation factor better than some constant $c \in (1,18/17)$, even for splitting parameter that are far from the submodular region where our best approximation factors can be arbitrarily bad. Another direction is to apply our approximation techniques to improve downstream hypergraph clustering problems where the most meaningful choice of cut function is only approximately submodular. 
      
	\bibliographystyle{plain}
	\bibliography{preprint.bib}
\end{document}